\documentclass[12pt,a4paper]{article}

\usepackage{fullpage}
\usepackage[utf8]{inputenc}
\usepackage{amssymb}
\usepackage{amsmath}
\usepackage{amsthm}
\usepackage{mathabx}
\usepackage[round]{natbib}
\usepackage{url}
\usepackage{graphicx}
\usepackage{multirow}
\usepackage{setspace}
\usepackage{enumerate}
\usepackage{color}
\usepackage{booktabs,rotating}
\usepackage[nolists]{endfloat}
\usepackage{endrotfloat}

\newcommand{\N}{\mathbb{N}}
\newcommand{\R}{\mathbb{R}}
\newcommand{\Rbar}{\overline{\R}}
\newcommand{\dd}{\mathrm{d}}
\newcommand{\G}{\mathbb{G}}
\newcommand{\D}{\mathbb{D}}
\newcommand{\K}{\mathbb{K}}
\renewcommand{\P}{\mathbb{P}}
\newcommand{\Z}{\mathbb{Z}}
\newcommand{\W}{\mathbb{W}}
\newcommand{\CC}{\mathcal{C}}
\newcommand{\FF}{\mathcal{F}}
\newcommand{\OO}{\mathcal{O}}
\newcommand{\HH}{\mathcal{H}}
\newcommand{\GG}{\mathcal{G}}
\newcommand{\Smc}{\mathcal{S}}
\newcommand{\cov}{\mathrm{cov}}

\newcommand{\1}{\mathbf{1}}
\newcommand{\ip}[1]{\lfloor #1 \rfloor}
\newcommand{\vep}{\varepsilon}

\renewcommand{\Pr}{\mathbf{P}}
\newcommand{\p}{\overset{\Pr^*}{\to}}
\newcommand{\as}{\overset{\mathrm{a.s.}}{\longrightarrow}}
\newcommand{\aso}{\overset{\mathrm{a.s.}*}{\longrightarrow}}

\newtheorem{prop}{Proposition}
\newtheorem{thm}{Theorem}

\newtheorem{lem}{Lemma}

\title{Nonparametric tests for change-point detection \`a la Gombay and Horv\'ath}

\author{
  Mark Holmes\\
  \small{Department of Statistics} \\
  \small{The University of Auckland} \\
  \small{Private Bag 92019, Auckland 1142, New Zealand} \\
  \small{\texttt{mholmes@stat.auckland.ac.nz}}
  \and
  Ivan Kojadinovic\\ 
  \small{Laboratoire de math\'ematiques et applications, UMR CNRS 5142} \\
  \small{Universit\'e de Pau et des Pays de l'Adour} \\
  \small{B.P. 1155, 64013 Pau Cedex, France} \\
  \small{\texttt{ivan.kojadinovic@univ-pau.fr}}
  \and
  Jean-Fran\c{c}ois Quessy\\ 
  \small{D\'epartement de math\'ematiques et d’informatique} \\
  \small{Universit\'e du Qu\'ebec \`a Trois-Rivi\`eres} \\
  \small{Trois-Rivi\`eres, Québec, C.P. 500, G9A 5H7 Canada} \\
  \small{\texttt{jean-francois.quessy@uqtr.ca}}
}

\begin{document}
\maketitle

\begin{abstract}
The nonparametric test for change-point detection proposed by Gombay and Horv\'ath is revisited and extended in the broader setting of empirical process theory. The resulting testing procedure for potentially multivariate observations is based on a sequential generalization of the functional multiplier central limit theorem and on modifications of Gombay and Horv\'ath's seminal approach that appears to improve the finite-sample behavior of the tests. A large number of candidate test statistics based on processes indexed by lower-left orthants and half-spaces are considered and their performance is studied through extensive Monte Carlo experiments involving univariate, bivariate and trivariate data sets. Finally, practical recommendations are provided and the tests are illustrated on trivariate hydrological data.

\medskip


\end{abstract}


\section{Introduction}

Let $X_1, \ldots, X_n$ be a sequence of independent $d$-dimensional random vectors for some fixed integer $d \geq 1$. The aim of this work is to study, both theoretically and empirically, nonparametric tests for the detection of a change-point in the sequence $X_1, \ldots, X_n$. The corresponding null hypothesis is
\begin{equation}
\label{H0}
H_0 : \,\exists \, P_0\mbox{ such that } X_1,\ldots,X_n \mbox{ have law } P_0.
\end{equation}
As frequently done, the behavior of the derived tests will be investigated under the alternative hypothesis of a single change-point:
\begin{align}
\nonumber 
H_1 :\, &\exists \mbox{ distinct } P_1 \mbox { and } P_2\mbox{, and } k^\star \in \{1, \ldots, n-1\} \mbox{ such that }\\
\label{H1} 
&X_1, \ldots, X_{k^\star} \mbox{ have law } P_1 \mbox{ and } X_{k^\star+1}, \ldots, X_n \mbox{ have law } P_2.
\end{align} 

There exists an abundant literature on nonparametric tests for change-point detection. We shall not review here procedures designed for serially dependent observations. The approaches proposed for sequences of independent observations differ, on one hand, according to the test statistic, and on the other hand, according to the resampling technique used to compute an approximate $p$-value for the test statistic. In terms of the test statistic, two frequently encountered classes of approaches are those based on $U$-statistics \citep[see e.g.][]{CsoHor88,Fer94,GomHor02,HorHus05} and those based on empirical c.d.f.s \citep[see e.g.][]{GomHor99,HorSha07}. As far as the resampling technique is concerned, one finds approaches based on permutations of the original sequence \cite[see e.g.][]{AntHus01,HorHus05,HorSha07} and approaches that use a weighted bootstrap based on multiplier central limit theorems \cite[see e.g.][]{GomHor99,GomHor02}. For a broader presentation of the field of change-point analysis, we refer the reader to the monographs by \cite{BroDar93} and \cite{CsoHor97}.

In this paper, we revisit and extend the approach proposed by \cite{GomHor99} based on the test statistic 
$$
T_{n,\vee} = \max_{1 \leq k \leq n-1} \frac{k(n-k)}{n^{3/2}} \sup_{x \in \R^d} \left| F_k(x) - F_{n-k}^\star(x) \right|,
$$
where 
$$
F_k(x) = \frac{1}{k} \sum_{i=1}^k \1(X_i \leq x) \qquad \mbox{and} \qquad F_{n-k}^\star(x) = \frac{1}{n-k} \sum_{i=k+1}^n \1(X_i \leq x), \qquad x \in \R^d,
$$
are the empirical c.d.f.s computed from $X_1,\dots,X_k$ and $X_{k+1},\dots,X_n$, respectively \cite[see also][Section 2.6]{CsoHor97}. From a theoretical perspective, we work in the framework of the theory of empirical processes as presented for instance in \cite{vanWel96} and \cite{Kos08}. To obtain results that are valid for many different {\em classes of functions} (in the sense of empirical process theory -- see Section \ref{notation}), we first extend the multiplier central limit theorem \citep[see e.g.][Theorem 10.1 and Corollary 10.3]{Kos08} to the sequential setting. This allows us to obtain interesting generalizations of Theorems~2.1,~2.2 and~2.3 of \cite{GomHor99}. In particular, we propose a slightly different {\em multiplier} process that appears to lead to better behaved tests in the case of moderate sample size. From a more practical perspective, we consider a large number of candidate test statistics based on processes indexed by {\em lower-left orthants} and by {\em half-spaces}, and we study the finite-sample performance of the corresponding tests through extensive Monte Carlo experiments involving univariate, bivariate and trivariate data sets. As we shall see, in the multivariate case, the tests based on processes indexed by half-spaces appear to be substantially more powerful than more classical tests based on multivariate empirical c.d.f.s (i.e., based on processes indexed by lower-left orthants).

The paper is organized as follows. In the second section, we state the theoretical results at the root of the studied class of tests in the broad setting of empirical process theory. The third section is devoted to an application of the theorems of Section~\ref{theory} to the derivation of nonparametric tests for change-point detection for two classes of functions which are the collection of indicator functions of lower-left orthants and the collection of indicator functions of half-spaces. The results of large-scale Monte-Carlo experiments comparing the finite-sample behavior of the tests are partially reported in the fourth section. The last section contains practical recommendations and presents an application of the studied tests to trivariate hydrological data. All the proofs are relegated to the appendices.

Note finally that the code of all the tests studied in this work will be documented and released as an R package whose tentative name is {\tt npcp}.

\section{Theoretical results for change-point detection}
\label{theory}

\subsection{Notation and setting}
\label{notation}

All the random variables used in this work are defined with respect to the underlying probability space $(\Omega,\GG,\Pr)$ and the outer probability measure corresponding to $\Pr$ is denoted by~$\Pr^*$. 

Let $X_1,\ldots,X_n$ be i.i.d.\ $d$-dimensional random vectors with law $P$, and let $\FF$ be a class of measurable functions from $\R^d$ to $\R$. The empirical measure is defined to be $\P_n = n^{-1} \sum_{i=1}^n \delta_{X_i}$, where $\delta_x$ is the measure that assigns a mass of 1 at $x$ and zero elsewhere. For $f\in \FF$, $\P_n f$ denotes the expectation of $f$ under $\P_n$, and $P f$ the expectation under $P$, i.e.,
$$
\P_n f = \frac{1}{n} \sum_{i=1}^n f(X_i) \qquad \mbox{and} \qquad Pf = \int f \dd P.
$$
The empirical process evaluated at $f$ is then defined as $\G_n f = \sqrt{n} (\P_n f - Pf)$.

Saying that $\FF$ is $P$-Donsker means that the sequence of processes $\{ \G_n f : f \in \FF\}$ converges weakly to a $P$-Brownian bridge $\{ \G_P f : f \in \FF\}$ in the space  $\ell^\infty(\FF)$ of bounded functions from $\FF$ to $\R$ equipped with the uniform metric in the sense of Definition~1.3.3 of \citet{vanWel96}. Following usual notational conventions, this weak convergence will simply be denoted by $\G_n \leadsto \G_P$ in $\ell^\infty(\FF)$. Furthermore, we say that $F_e:\R^d \rightarrow \R$ is an envelope for $\FF$ if $F_e$ is measurable and $|f(x)|\le F_e(x)$ for every $f\in \FF$ and $x\in \R^d$.

The advantage of working in this general framework is that the forthcoming results remain valid for many $P$-Donsker classes $\FF$. By taking $\FF$ to be the class of indicator functions of lower-left orthants in $\R^d$, i.e., $\FF = \{y \mapsto \1(y \leq x) : x \in \Rbar^d\}$ with $\Rbar = \R \cup \{-\infty,\infty\}$, one recovers the setting studied in \citet[Section 2.6]{CsoHor97} and based on empirical cumulative distribution functions (c.d.f.s). Although this is a natural choice for $\FF$, many other choices might be of interest in practice such as the class of indicator functions of closed balls, rectangles or half-spaces \citep[see][for a related discussion regarding the choice of $\FF$]{Rom88}.

\subsection{A multiplier central limit theorem for the sequential empirical process}

The sequential empirical process is defined as
$$
\Z_n(s,f) = \frac{1}{\sqrt{n}} \sum_{i=1}^{\ip{ns}} \{f(X_i) - Pf \} = \sqrt{\lambda_n(s)} \G_{\ip{ns}} f, \qquad s \in [0,1], f \in \FF,
$$
where $\lambda_n(s) = \ip{ns} / n$ and with the convention that $\P_0 f = 0$ for all $f \in \FF$.

According to Theorem 2.12.1 of \cite{vanWel96}, $\FF$ being $P$-Donsker is equivalent to $\Z_n \leadsto \Z_P$ in $\ell^\infty([0,1] \times \FF)$, where $\Z_P$ is a tight centered mean-zero Gaussian process with covariance function
$$
\cov\{\Z_P(s,f), \Z_P(t,g)\} = (s \wedge t) (Pfg - Pf Pg)
$$
known as a {\em $P$-Kiefer-M\"uller} process.

Given i.i.d.\ random variables $\xi_1,\dots,\xi_n$ with mean 0 and variance 1, satisfying $\int_0^\infty \{ \Pr(|\xi_1| > x) \}^{1/2} \dd x < \infty$, and independent of the random sample $X_1,\dots,X_n$, we define the following {\em multiplier} version of $\Z_n$:
$$
\widetilde{\Z}_n(s,f) = \frac{1}{\sqrt{n}} \sum_{i=1}^{\ip{ns}} \xi_i \{f(X_i) - Pf \}, \qquad s \in [0,1], f \in \FF.
$$
Notice that the empirical process $\widetilde{\Z}_n$ depends on the unknown map $f \mapsto Pf$ and therefore cannot be computed. With applications in mind, we define two versions of $\widetilde{\Z}_n$ (depending on how $f \mapsto Pf$ is estimated) that can be fully computed. For any $s \in [0,1], f \in \FF
$, let
$$
\widehat{\Z}_n(s,f) = \frac{1}{\sqrt{n}} \sum_{i=1}^{\ip{ns}} \xi_i \{f(X_i) - \P_{\ip{ns}} f \} = \frac{1}{\sqrt{n}} \sum_{i=1}^{\ip{ns}} (\xi_i - \bar \xi_{\ip{ns}}) f(X_i),
$$
where $\bar \xi_{\ip{ns}} = \ip{ns}^{-1} \sum_{i=1}^{\ip{ns}} \xi_i$ and $\bar \xi_0 = 0$ by convention, and let
$$
\widecheck{\Z}_n(s,f) = \frac{1}{\sqrt{n}} \sum_{i=1}^{\ip{ns}} \xi_i \{f(X_i) - \P_n f \}.
$$

The following result is then a partial extension of the multiplier central limit theorem \citep[see e.g.][Theorem 10.1 and Corollary 10.3]{Kos08} to the sequential setting.

\begin{thm}
\label{multiplier}
Let $\FF$ be a $P$-Donsker class with measurable envelope $F_e$ such that $P F_e^2 < \infty$. Then, $(\Z_n, \widetilde{\Z}_n, \widehat{\Z}_n, \widecheck{\Z}_n ) \leadsto (\Z_P,\Z'_P,\Z'_P,\Z'_P )$ in $\{ \ell^\infty([0,1] \times \FF) \}^4$, where $\Z_P'$ is an independent copy of $\Z_P$.
\end{thm}

Theorem~\ref{multiplier} suggests the following interpretation: when $n$ is large, $\widetilde{\Z}_n$ can be regarded as ``almost'' an independent copy of $\Z_n$, while $\widehat{\Z}_n$ and $\widecheck{\Z}_n$ can be regarded as computable copies of $\widetilde{\Z}_n$. As we shall see, this interpretation is at the root of the resampling technique considered in Section~\ref{tests}. 

Although each of $\widehat{\Z}_n$ and $\widecheck{\Z}_n$ could be regarded as ``almost'' an independent copy of~$\Z_n$, their behavior for moderate $n$ might differ quite substantially. In Section~\ref{sims}, we empirically investigate which of $\widehat{\Z}_n$ or $\widecheck{\Z}_n$ leads to tests for change-point detection with the best finite-sample properties.

\subsection{Application to change-point detection}

Recall that the null and alternative hypotheses under consideration are given in~(\ref{H0}) and~(\ref{H1}), respectively.

Let $\FF$ be a class of measurable functions. In order to test the aforementioned hypotheses, we extend the approach studied in detail by \citet[Section 2.6]{CsoHor97} and compare, for all $k \in \{1, \dots, n-1\}$, 
$$
\P_k f = \frac{1}{k} \sum_{i=1}^k f(X_i) \qquad \mbox{and} \qquad \P_{n-k}^\star f = \frac{1}{n-k} \sum_{i=k+1}^n f(X_i), \qquad f \in \FF.
$$
Analogous to \citet[Section 2.6]{CsoHor97}, we define the process
$$
\D_n(s,f) = \sqrt{n} \, \lambda_n(s) \left\{ 1 - \lambda_n(s) \right\} \left(
\P_{\ip{ns}}f - \P_{n-\ip{ns}}^\star f  \right), \qquad s \in [0,1], f \in \FF,
$$
where $\lambda_n(s) = \ip{ns} / n$ and with the convention that $\P_0 f = 0$ and $\P_0^\star f = 0$ for all $f \in \FF$. Notice that, under the null hypothesis, for any $s \in [0,1]$ and $f \in \FF$, we have
\begin{equation}
\label{DnH0}
\D_n(s,f) = \left\{ 1 - \lambda_n(s) \right\} \Z_n(s,f) - \lambda_n(s) \{ \Z_n(1,f) - \Z_n(s,f) \} = \Z_n(s,f) - \lambda_n(s) \Z_n(1,f).
\end{equation}

With resampling in mind, we define two {\em multiplier} versions of $\D_n$ based on the multiplier versions of $\Z_n$ defined in the previous subsection. For any $s \in [0,1]$ and $f \in \FF$, 
let
$$
\widecheck{\D}_n(s,f) = \{ 1 - \lambda_n(s) \} \widecheck{\Z}_n(s,f) - \lambda_n(s) \{ \widecheck{\Z}_n(1,f) - \widecheck{\Z}_n(s,f) \} = \widecheck{\Z}_n(s,f) - \lambda_n(s) \widecheck{\Z}_n(1,f),
$$
and, following \cite{GomHor99}, let
\begin{equation}
\label{hatDn}
\widehat{\D}_n(s,f) = \{ 1 - \lambda_n(s) \} \widehat{\Z}_n(s,f) - \lambda_n(s) \widehat{\Z}_n^\star(s,f), 
\end{equation}
where 
\begin{equation}
\label{hatZnstar}
\widehat{\Z}_n^\star(s,f) = \frac{1}{\sqrt{n}} \sum_{i=\ip{ns}+1}^n (\xi_i - \bar \xi_{n-\ip{ns}}^\star) f(X_i) \qquad \mbox{with} \qquad \bar \xi_{n-\ip{ns}}^\star = \frac{1}{n-\ip{ns}} \sum_{i=\ip{ns}+1}^n \xi_i, 
\end{equation}
and $\bar \xi_0^\star = 0$ by convention. Notice that the process $\widehat{\Z}_n^\star$ defined above is, up to a small error term vanishing as $n \to\infty$, the version of the process $(s,f) \mapsto \widehat{\Z}_n(1-s,f)$ computed from the ``reversed'' sequence $(\xi_n,X_n),(\xi_{n-1},X_{n-1}),\dots,(\xi_1,X_1)$.

The following two results extend Theorems~2.1,~2.2 and~2.3 of \cite{GomHor99} and suggest, for large $n$ and under the null hypothesis, to interpret each of $\widehat{\D}_n$ and $\widecheck{\D}_n$  as an ``almost'' independent copy of $\D_n$.

\begin{thm}
\label{thm_H0}
Assume that $H_0$ holds and that $\FF$ is $P_0$-Donsker with measurable envelope $F_e$ such that $P_0 F_e^2 < \infty$. Then, $( \D_n,\widehat{\D}_n,\widecheck{\D}_n) \leadsto ( \D_{P_0},\D'_{P_0},\D'_{P_0} )$ in $\{ \ell^\infty([0,1] \times \FF) \}^3$, where $\D_{P_0}(s,f) = \Z_{P_0}(s,f) - s \Z_{P_0}(1,f)$, $s \in [0,1]$, $f \in \FF$, and $\D'_{P_0}$ is an independent copy of $\D_{P_0}$.
\end{thm}

As we continue, for any $J:\FF \to \R$, $\sup_{f \in \FF} | Jf |$ will be denoted by $\|J\|_\FF$. Also, for any sequence of maps $Y_1,Y_2,\dots$, each from $\Omega$ to $\R$, we say that the sequence $Y_n$ is bounded in outer probability if, for any $\vep > 0$, there exists $M > 0$ such that $\sup_{n \in \N} \Pr^* \left( |Y_n| > M \right) < \vep$. 

\begin{thm}
\label{thm_H1}
Assume that $H_1$ holds with $k^\star = \ip{nt}$ for some $t \in (0,1)$ and that $\FF$ is $P_1$ and $P_2$-Donsker with measurable envelope $F_e$ such that $P_1 F_e^2 < \infty$ and $P_2 F_e^2 < \infty$. Then,
\begin{enumerate}[(i)]
\item 
$\sup_{s \in [0,1]} \| n^{-1/2} \D_n(s,f) - K_t(s,f) \|_\FF \p 0$, \\
where 
$K_t(s,f) = ( P_1f - P_2 f ) ( s \wedge t) \{ 1-  (s \vee t)\}$,
\item $\sup_{s \in [0,1]} \| \widehat{\D}_n(s,f) \|_\FF$ is bounded in outer probability,
\item $\widecheck{\D}_n$ converges weakly in $\ell^\infty([0,1] \times \FF)$.
\end{enumerate}
\end{thm}

The previous result will be used in the next section to show that various tests for change-point detection based on $\D_n$ will tend to reject $H_0$ under $H_1$ as $n$ increases.

\section{Tests for change-point detection \`a la Gombay and Horv\'ath}
\label{tests}

The aim of this section is to use the results of the previous section to derive tests for change-point detection in the spirit of those proposed by \cite{GomHor99}. Among the many possible choices for $\FF$, we consider the following two: 
\begin{description}
\item[($\CC 1$)] the collection $\OO$ of indicator functions of lower-left orthants in $\R^d$, where 
$$
\OO = \{f_x(y) = \1(y \leq x) : x \in \Rbar^d\};
$$
\item[($\CC 2$)] the collection $\HH$ of indicator functions of half-spaces in $\R^d$, where 
$$
\HH = \{f_{a,b}(y) = \1(a^\top y \leq b) : a \in \Smc_d, b \in \Rbar \}
$$ 
and $\Smc_d$ is the subset of $\R^d$ composed of vectors with Euclidean norm one.
\end{description}
It is well-known that lower-left orthants and half-spaces are Vapnik-\v Chervonenkis classes of sets. Consequently, $\OO$ and $\HH$ are $P$-Donsker for any law $P$ \citep[see e.g.][]{vanWel96,Kos08}.

As we continue, in the case of choice~($\CC 1$), for any $s \in [0,1]$ and $f_x \in \OO$, $\D_n(s,f_x)$, $\widehat{\D}_n(s,f_x)$ and $\widecheck{\D}_n(s,f_x)$ will simply be denoted by $\D_n(s,x)$, $\widehat{\D}_n(s,x)$ and $\widecheck{\D}_n(s,x)$, respectively. 
Similarly, in the case of choice~($\CC 2$), for any $s \in [0,1]$ and $f_{a,b} \in \HH$, $\D_n(s,f_{a,b})$, $\widehat{\D}_n(s,f_{a,b})$ and $\widecheck{\D}_n(s,f_{a,b})$ will be denoted by $\D_n(s,a,b)$, $\widehat{\D}_n(s,a,b)$ and $\widecheck{\D}_n(s,a,b)$, respectively.

In the framework under consideration, a change in the sequence $X_1,\dots,X_n$ can occur at any point $k \in \{1,\dots,n-1\}$. A test for change-point detection could therefore be obtained by first defining a test statistic for any possible change-point $k \in \{1,\dots,n-1\}$, and then by combining the resulting $n-1$ statistics into a global statistic using some function from $\psi:\R^{n-1} \to \R$.

In the case of choice ($\CC 1$), two natural possibilities for the $n-1$ change-point statistics are respectively
$$
S_{n,k} = \int_{\R^d} \left\{ \D_n \left( \frac{k}{n},x \right) \right\}^2 \dd F_n(x) = \frac{1}{n} \sum_{i=1}^n \left\{ \D_n \left( \frac{k}{n},X_i \right) \right\}^2,  \qquad k \in \{1,\dots,n-1\},
$$
where $F_n(x)=\P_nf_x$, $x \in \Rbar^d$, is the empirical c.d.f.\ computed from $X_1,\dots,X_n$, and
$$
T_{n,k} = \sup_{x \in \R^d} \left|  \D_n \left( \frac{k}{n},x \right) \right| = \max_{1 \leq i \leq n} \left|  \D_n \left( \frac{k}{n},X_i \right) \right|, \qquad k \in \{1,\dots,n-1\}.
$$
Two natural choices for the function $\psi$ are the maximum and the arithmetic mean which leads to the following four global statistics:
\begin{gather*}
S_{n,\vee} = \max_{1 \leq k \leq n-1} S_{n,k} = \sup_{s \in [0,1]} \int_{\R^d} \left\{ \D_n \left( s,x \right) \right\}^2 \dd F_n(x), \\
T_{n,\vee} = \max_{1 \leq k \leq n-1} T_{n,k} = \sup_{s \in [0,1]} \sup_{x \in \R^d} \left|  \D_n \left( s,x \right) \right|,\\
S_{n,+} = \frac{1}{n} \sum_{k=1}^{n-1} S_{n,k} = \int_0^1 \int_{\R^d} \left\{ \D_n \left( s,x \right) \right\}^2 \dd F_n(x) \dd s, \\
T_{n,+} = \frac{1}{n} \sum_{k=1}^{n-1} T_{n,k} = \int_0^1 \sup_{x \in \R^d} \left|  \D_n \left( s,x \right) \right| \dd s.
\end{gather*}
Note that $T_{n,\vee}$ is the global statistic used in \cite{GomHor99}.

In the case of choice ($\CC 2$), for any $k \in \{1,\dots,n-1\}$, we first define
$$
U_{n,k} = \int_{\Smc_d^+} \int_{\R} \left\{ \D_n \left( \frac{k}{n},a,b \right) \right\}^2 \dd F_{a,n}(b) \dd a = \int_{\Smc_d^+} \frac{1}{n} \sum_{i=1}^n \left\{ \D_n \left( \frac{k}{n},a,a^\top X_i \right) \right\}^2 \dd a, 
$$
where $\Smc_d^+ = \{a \in \Smc_d : a_1 > 0\}$ and, for any $a \in \Smc_d^+$, $F_{a,n}$ is the (univariate) empirical c.d.f.\ computed from the projected sample $a^\top X_1,\dots,a^\top X_n$, and
$$
V_{n,k} = \sup_{a \in \Smc_d^+} \sup_{b \in \R} \left|  \D_n \left( \frac{k}{n},a,b \right) \right| = \sup_{a \in \Smc_d^+} \max_{1 \leq i \leq n} \left|  \D_n \left( \frac{k}{n},a,a^\top X_i \right) \right|.
$$
As in the case of choice ($\CC 1$), this leads to four global statistics depending on whether the change-point statistics are combined using the maximum or the arithmetic mean, i.e.,
\begin{gather*}
U_{n,\vee} = \max_{1 \leq k \leq n-1} U_{n,k} = \sup_{s \in [0,1]} \int_{\Smc_d^+} \int_{\R} \left\{ \D_n \left( s,a,b \right) \right\}^2 \dd F_{a,n}(b) \dd a, \\
V_{n,\vee} = \max_{1 \leq k \leq n-1} V_{n,k} = \sup_{s \in [0,1]} \sup_{a \in \Smc_d^+} \sup_{b \in \R} \left|  \D_n \left( s,a,b \right) \right|,\\
U_{n,+} = \frac{1}{n} \sum_{k=1}^{n-1} U_{n,k} = \int_0^1 \int_{\Smc_d^+} \int_{\R} \left\{ \D_n \left( s,a,b \right) \right\}^2 \dd F_{a,n}(b) \dd a \dd s, \\
V_{n,+} = \frac{1}{n} \sum_{k=1}^{n-1} V_{n,k} = \int_0^1 \sup_{a \in \Smc_d^+} \sup_{b \in \R} \left|  \D_n \left( s,a,b \right) \right| \dd s.
\end{gather*}
In our Monte Carlo experiments, the integral and the supremum over $a \in \Smc_d^+$ in the definitions of $U_{n,k}$ and $V_{n,k}$, respectively, were approximated numerically based on a uniform discretization of $\Smc_d^+$ into $m$ points. The implementation of the tests based on $U_{n,k}$ and $V_{n,k}$ is discussed in more detail in Appendix~\ref{implementation}. Notice finally that the change-point statistics $S_{n,k}$ and $U_{n,k}$ (resp.\ $T_{n,k}$ and $V_{n,k}$) coincide when $d=1$ since $\Smc_1^+ = \{1\}$. 

Let us now explain how approximate $p$-values for these statistics can be computed using the multiplier processes $\widehat{\D}_n$ and $\widecheck{\D}_n$. For the sake of brevity, we present the approach and state the key results only for $S_{n,\vee}$ as the cases of the other seven global statistics are similar. 

Let $N$ be a large integer and let $\xi_i^{(j)}$, $i \in \{1,\dots,n\}$, $j \in \{1,\dots,N\}$, be i.i.d.\ random variables with mean 0 and variance 1 satisfying $\int_0^\infty \{ \Pr(|\xi_i^{(j)}| > x) \}^{1/2} \dd x < \infty$, and independent of $X_1,\dots,X_n$. Also, for any $j \in \{1,\dots,N\}$, let $\widehat{\D}_n^{(j)}$ (resp.\ $\widecheck{\D}_n^{(j)}$) denote the version of $\widehat{\D}_n$ (resp.\ $\widecheck{\D}_n$) computed from $\xi_1^{(j)},\dots,\xi_n^{(j)}$. Moreover, for any $j \in \{1,\dots,N\}$, let
$$
\widehat{S}_{n,\vee}^{(j)} = \sup_{s \in [0,1]} \int_{\R^d} \left\{ \widehat{\D}_n^{(j)} \left( s,x \right) \right\}^2 \dd F_n(x)  \qquad \mbox{and} \qquad \widecheck{S}_{n,\vee}^{(j)} = \sup_{s \in [0,1]}  \int_{\R^d} \left\{ \widecheck{\D}_n^{(j)} \left( s,x \right) \right\}^2 \dd F_n(x).
$$
The following result is then essentially a corollary of Theorem~\ref{thm_H0}.

\begin{prop}
\label{prop_S_T_H0}
Under $H_0$, 
$$
\left( S_{n,\vee},\widehat{S}_{n,\vee}^{(1)},\dots,\widehat{S}_{n,\vee}^{(N)},\widecheck{S}_{n,\vee}^{(1)},\dots,\widecheck{S}_{n,\vee}^{(N)} \right) \leadsto \left( S_{\vee},S^{(1)}_{\vee},\dots,S^{(N)}_{\vee},S^{(1)}_{\vee},\dots,S^{(N)}_{\vee} \right)
$$
in $[0,\infty)^{(2N+1)}$, where
$$
S_{\vee} = \sup_{s \in [0,1]} \int_{\R^d} \{ \D_{P_0}(s,x) \}^2 \dd F_0(x) 
$$
is the weak limit of $S_{n,\vee}$, $F_0$ is the c.d.f.\ corresponding to $P_0$, and $S^{(1)}_{\vee},\dots,S^{(N)}_{\vee}$ are independent copies of $S_{\vee}$.
\end{prop}

The previous proposition suggests interpreting the $\widehat{S}_{n,\vee}^{(j)}$ (resp.\ the $\widecheck{S}_{n,\vee}^{(j)}$) under the null hypothesis as $N$ ``almost'' independent copies of $S_{n,\vee}$ and thus computing an approximate $p$-value for $S_{n,\vee}$ as 
\begin{equation}
\label{pvalues}
\frac{1}{N} \sum_{j=1}^N \1 \left( \widehat{S}_{n,\vee}^{(j)} \geq S_{n,\vee} \right) \qquad \mbox{or as} \qquad \frac{1}{N} \sum_{j=1}^N \1 \left( \widecheck{S}_{n,\vee}^{(j)} \geq S_{n,\vee}   \right).
\end{equation}

\begin{prop}
\label{prop_S_T_H1}
Assume that $H_1$ holds with $k^\star = \ip{nt}$ for some $t \in (0,1)$. Then, $S_{n,\vee} \p + \infty$ while, for any $j \in \{1,\dots,N\}$, $\widehat{S}_{n,\vee}^{(j)}$ and $\widecheck{S}_{n,\vee}^{(j)}$ are bounded in outer probability. 
\end{prop}

A consequence of the previous proposition is that, under $H_1$, the approximate $p$-values for $S_{n,\vee}$ will tend to zero in outer probability. As mentioned earlier, results analogous to Propositions~\ref{prop_S_T_H0} and~\ref{prop_S_T_H1} can be obtained for $S_{n,+}$, $T_{n,\vee}$, $T_{n,+}$, $U_{n,\vee}$, $U_{n,+}$, $V_{n,\vee}$ and $V_{n,+}$.

\section{Monte Carlo experiments}
\label{sims}

Large-scale Monte Carlo experiments were carried out in order to study the finite-sample performance of the tests defined in the previous section. Let $Q_n$ be one of $S_{n,\vee}$, $S_{n,+}$, $T_{n,\vee}$, $T_{n,+}$, $U_{n,\vee}$, $U_{n,+}$, $V_{n,\vee}$ and $V_{n,+}$. In the rest of the paper, the test based on $Q_n$ will be referred to as {\em the test based on $\widehat{Q}_n$ (resp.\ $\widecheck{Q}_n$)} when its approximate $p$-value is computed using the multiplier processes $\widehat{\D}_n^{(j)}$ (resp.\ $\widecheck{\D}_n^{(j)}$). 

To compare the power of the aforementioned tests, several univariate and multivariate scenarios were considered and 1000 samples of size $n$ were generated under each scenario for $n \in \{50,100,200\}$. Recall that the c.d.f.\ corresponding to $P_0$ in $H_0$ defined in~(\ref{H0}) is denoted by $F_0$. Similarly, the distinct c.d.f.s corresponding to $P_1$ and $P_2$ in $H_1$ defined in~(\ref{H1}) will be denoted by $F_1$ and $F_2$, respectively, as we continue. In all scenarios, the multipliers appearing in the $\widehat{\D}_n^{(j)}$ and $\widecheck{\D}_n^{(j)}$ were taken from the standard normal distribution. All approximate $p$-values were computed from $N=1000$ multiplier realizations and the tests were carried out at the 5\% level of significance. 

\begin{table}[ht]
\begin{center}
\caption{Percentage of rejection of $H_0$ in the univariate case computed from 1000 random samples of size $n$ generated under $H_0$ defined in~(\ref{H0}) where $F_0$ is the c.d.f. of the standard normal distribution.}
\label{H0univ}
\begin{tabular}{rrrrrrrrrrrrr}
  \hline
  $n$ & $\widehat{S}_{n,\vee}$ & $\widecheck{S}_{n,\vee}$ & $S_{n,\vee}^*$ & $\widehat{S}_{n,+}$ & $\widecheck{S}_{n,+}$ & $S_{n,+}^*$ & $\widehat{T}_{n,\vee}$ & $\widecheck{T}_{n,\vee}$ & $T_{n,\vee}^*$ & $\widehat{T}_{n,+}$ & $\widecheck{T}_{n,+}$ & $T_{n,+}^*$ \\ \hline
50 & 7.2 & 5.7 & 5.7 & 7.7 & 5.1 & 5.9 & 6.7 & 5.8 & 5.7 & 8.4 & 5.2 & 4.3 \\ 
  100 & 6.5 & 5.5 & 6.2 & 6.0 & 4.9 & 6.1 & 7.1 & 6.6 & 6.5 & 8.1 & 6.2 & 6.2 \\ 
  200 & 5.9 & 4.8 & 4.3 & 5.9 & 4.6 & 5.5 & 5.4 & 4.5 & 4.5 & 7.5 & 5.4 & 4.3 \\ 
   \hline
\end{tabular}
\end{center}
\end{table}

From the previous section, it is easy to verify that, in the univariate case, the change-point statistics $U_{n,k}$ (resp.\ $V_{n,k}$) coincide with the $S_{n,k}$ (resp.\ $T_{n,k}$) since $\Smc_1^+ = \{1\}$, and that the $S_{n,k}$ and the $T_{n,k}$ are solely based on ranks. From the latter fact, it follows that, under $H_0$ and the assumption of continuity of $F_0$, independent realizations of each of the four global statistics based on the $S_{n,k}$ or the $T_{n,k}$ can be obtained by computing these global statistics from independent samples of size $n$ generated from the standard uniform distribution. A natural way of computing an approximate $p$-value for each of the four global statistics then consists of proceeding analogously to~(\ref{pvalues}) using $N$ independent realizations. As we continue, the resulting four univariate tests will be referred to as {\em the tests based on $S^*_{n,\vee}$, $S^*_{n,+}$, $T^*_{n,\vee}$ and $T_{n,+}^*$}. 

Table~\ref{H0univ} gives rejection percentages of $H_0$ in dimension one for all the aforementioned versions of the tests when data are generated under $H_0$. As can be seen, the tests whose approximate $p$-value is computed using the processes $\widehat{\D}_n^{(j)}$ seem to be too liberal (at least for $n \in \{50,100\}$), and more liberal than their version computed from the $\widecheck{\D}_n^{(j)}$. Nevertheless, as expected, the empirical levels of the multiplier tests improve as $n$ increases in the sense that they become closer to the 5\% nominal level. Note that the tests based on $S^*_{n,\vee}$, $S^*_{n,+}$, $T^*_{n,\vee}$ and $T_{n,+}^*$ provide a sort of benchmark as, by construction, they should hold their level well for any $n$ provided $N$ is taken sufficiently large. 

\begin{sidewaystable}[ht]
\begin{center}
\caption{Percentage of rejection of $H_0$ in the univariate case computed from 1000 samples of size $n$ generated under $H_1$ defined in~(\ref{H1}),  where $k^\star = \ip{nt}$, and $F_1$ and $F_2$ are the c.d.f.s of the distributions given in the first two columns.}
\label{H1univ}
\begin{tabular}{llrrrrrrrrrrrrrr}
  \hline
  $F_1$ & $F_2$ & $n$ & $t$ & $\widehat{S}_{n,\vee}$ & $\widecheck{S}_{n,\vee}$ & $S_{n,\vee}^*$ & $\widehat{S}_{n,+}$ & $\widecheck{S}_{n,+}$ & $S_{n,+}^*$ & $\widehat{T}_{n,\vee}$ & $\widecheck{T}_{n,\vee}$ & $T_{n,\vee}^*$ & $\widehat{T}_{n,+}$ & $\widecheck{T}_{n,+}$ & $T_{n,+}^*$ \\ \hline
N(0,1) & N(0.5,1) & 50 & 0.10 & 9.1 & 7.1 & 6.8 & 10.1 & 7.1 & 8.2 & 9.0 & 7.3 & 6.9 & 14.2 & 8.3 & 8.6 \\ 
  N(0,1) & N(0.5,1) & 50 & 0.25 & 18.4 & 15.9 & 13.4 & 19.1 & 14.4 & 14.4 & 17.0 & 13.6 & 8.9 & 23.0 & 13.6 & 12.0 \\ 
  N(0,1) & N(0.5,1) & 50 & 0.50 & 34.0 & 30.5 & 31.6 & 34.1 & 29.0 & 32.1 & 32.6 & 29.5 & 25.6 & 35.4 & 26.9 & 24.9 \\ 
  N(0,1) & N(0.5,1) & 100 & 0.10 & 9.7 & 8.2 & 6.5 & 10.1 & 8.5 & 7.3 & 9.7 & 8.7 & 6.7 & 11.6 & 9.4 & 7.3 \\ 
  N(0,1) & N(0.5,1) & 100 & 0.25 & 36.9 & 34.0 & 33.8 & 36.2 & 33.3 & 31.7 & 30.4 & 28.6 & 25.1 & 37.2 & 31.4 & 29.3 \\ 
  N(0,1) & N(0.5,1) & 100 & 0.50 & 58.7 & 55.9 & 54.1 & 56.7 & 53.9 & 56.0 & 52.1 & 48.6 & 42.1 & 53.6 & 49.3 & 45.5 \\ 
  N(0,1) & N(0.5,1) & 200 & 0.10 & 15.9 & 15.1 & 14.0 & 18.6 & 17.3 & 16.8 & 15.3 & 14.6 & 11.1 & 22.1 & 19.5 & 18.0 \\ 
  N(0,1) & N(0.5,1) & 200 & 0.25 & 65.2 & 64.1 & 63.9 & 65.8 & 64.3 & 62.5 & 56.7 & 55.5 & 54.8 & 61.6 & 57.2 & 52.2 \\ 
  N(0,1) & N(0.5,1) & 200 & 0.50 & 87.1 & 86.4 & 86.8 & 86.1 & 85.3 & 85.0 & 81.2 & 80.6 & 79.1 & 81.2 & 79.5 & 76.2 \\ 
  N(0,1) & N(0,2) & 50 & 0.10 & 6.8 & 5.6 & 5.9 & 8.1 & 5.4 & 4.9 & 8.1 & 6.0 & 5.1 & 11.9 & 7.0 & 6.3 \\ 
  N(0,1) & N(0,2) & 50 & 0.25 & 8.9 & 6.4 & 7.1 & 11.7 & 9.4 & 6.9 & 11.3 & 10.0 & 8.1 & 17.0 & 10.9 & 8.4 \\ 
  N(0,1) & N(0,2) & 50 & 0.50 & 12.9 & 10.3 & 9.9 & 18.8 & 13.5 & 14.1 & 19.6 & 16.6 & 12.6 & 27.8 & 18.4 & 17.4 \\ 
  N(0,1) & N(0,2) & 100 & 0.10 & 7.1 & 6.4 & 6.6 & 8.1 & 6.8 & 6.3 & 8.6 & 7.7 & 6.0 & 10.7 & 8.0 & 6.0 \\ 
  N(0,1) & N(0,2) & 100 & 0.25 & 8.8 & 7.5 & 8.1 & 17.5 & 14.3 & 15.4 & 14.2 & 13.2 & 12.1 & 25.3 & 19.2 & 18.8 \\ 
  N(0,1) & N(0,2) & 100 & 0.50 & 23.1 & 20.8 & 24.9 & 38.4 & 34.2 & 39.1 & 34.5 & 31.9 & 27.5 & 46.6 & 38.9 & 37.7 \\ 
  N(0,1) & N(0,2) & 200 & 0.10 & 6.4 & 5.9 & 5.9 & 9.1 & 8.2 & 6.9 & 9.0 & 8.7 & 6.7 & 14.5 & 12.2 & 8.6 \\ 
  N(0,1) & N(0,2) & 200 & 0.25 & 21.0 & 19.6 & 21.2 & 50.2 & 47.8 & 41.3 & 33.6 & 32.6 & 32.5 & 54.7 & 49.3 & 43.8 \\ 
  N(0,1) & N(0,2) & 200 & 0.50 & 64.8 & 64.0 & 65.1 & 82.8 & 81.7 & 82.3 & 72.1 & 71.0 & 65.7 & 83.7 & 82.5 & 79.1 \\ 
  E(1) & E(0.5) & 50 & 0.10 & 8.1 & 6.8 & 6.0 & 9.3 & 7.2 & 7.3 & 8.9 & 7.4 & 5.8 & 14.5 & 8.7 & 7.8 \\ 
  E(1) & E(0.5) & 50 & 0.25 & 28.0 & 25.3 & 22.1 & 30.7 & 25.5 & 22.6 & 25.5 & 22.4 & 19.6 & 33.8 & 24.2 & 22.3 \\ 
  E(1) & E(0.5) & 50 & 0.50 & 50.5 & 46.2 & 46.3 & 49.6 & 43.8 & 48.8 & 45.0 & 40.8 & 40.7 & 48.6 & 40.7 & 42.9 \\ 
  E(1) & E(0.5) & 100 & 0.10 & 12.7 & 11.0 & 7.9 & 14.3 & 11.7 & 11.4 & 12.8 & 11.6 & 11.1 & 16.7 & 13.1 & 10.9 \\ 
  E(1) & E(0.5) & 100 & 0.25 & 51.9 & 50.0 & 50.8 & 52.2 & 48.4 & 47.4 & 45.2 & 42.9 & 41.3 & 49.9 & 44.9 & 45.3 \\ 
  E(1) & E(0.5) & 100 & 0.50 & 77.4 & 76.1 & 74.6 & 77.3 & 74.8 & 72.9 & 72.4 & 70.4 & 70.7 & 73.8 & 69.4 & 67.9 \\ 
  E(1) & E(0.5) & 200 & 0.10 & 19.9 & 19.3 & 20.2 & 24.9 & 23.2 & 23.8 & 16.8 & 16.3 & 13.9 & 25.4 & 23.3 & 19.4 \\ 
  E(1) & E(0.5) & 200 & 0.25 & 87.0 & 86.2 & 85.5 & 84.4 & 83.1 & 81.4 & 79.0 & 78.0 & 76.2 & 80.8 & 77.3 & 77.9 \\ 
  E(1) & E(0.5) & 200 & 0.50 & 96.6 & 96.3 & 95.7 & 95.4 & 94.9 & 95.7 & 94.9 & 95.0 & 94.7 & 93.2 & 92.7 & 91.7 \\ 
   \hline
\end{tabular}
\end{center}
\end{sidewaystable}

Rejection percentages of $H_0$ in the univariate case when data are generated under $H_1$ are reported in Table~\ref{H1univ}. Three scenarios are considered: $F_1$ and $F_2$ are the c.d.f.s of the $N(0,1)$ and the $N(0.5,1)$ distributions, respectively; $F_1$ and $F_2$ are the c.d.f.s of the $N(0,1)$ and the $N(0,2)$ distributions, respectively; and $F_1$ and $F_2$ are the c.d.f.s of the exponential $E(1)$ and $E(0.5)$ distributions, respectively. Notice that the parameter $t$ taking its values in $\{0.1,0.25,0.5\}$ determines the position of the change-point in $H_1$ as $k^\star = \ip{nt}$. As one can see, the tests based on the processes $\widehat{\D}_n^{(j)}$ are consistently slightly more powerful than those based on the $\widecheck{\D}_n^{(j)}$, while the rejections rates of the latter are, overall, comparable to those of the tests based on simulation from the standard uniform distribution. This merely appears to be due to the fact that the tests based on the processes $\widehat{\D}_n^{(j)}$ are slightly too liberal for the sample sizes under consideration. Notice that the differences in power decrease as $n$ increases, as expected. Among the tests based on the $\widecheck{\D}_n^{(j)}$, the one based on $\widecheck{S}_{n,\vee}$ seems, overall, to be the best choice for detecting changes in mean, while the test based on $\widecheck{T}_{n,+}$ seems, overall, to be the best choice for detecting changes in variance. The former seems also to be the best choice, overall, when data are generated under the third scenario involving exponential distributions. If one is willing to make continuity assumptions on the underlying distributions, the tests based on $S^*_{n,\vee}$ and $T^*_{n,+}$ are equivalently good candidates. Clearly, there exists more powerful test for change-point detection if it is assumed that only a change in mean or variance can occur \citep[see e.g.][]{BroDar93,CsoHor97}.

In dimension two and three, we considered multivariate distributions constructed from \citet{Skl59}'s representation theorem. The latter result states that any multivariate c.d.f.\ $F:\R^d \to [0,1]$ whose marginal c.d.f.s $F^{[1]},\dots,F^{[d]}$ are continuous can be expressed in terms of a unique $d$-dimensional copula $C$ as
$$
F(x) = C\{F^{[1]}(x_1),\dots,F^{[d]}(x_d)\}, \qquad x \in \R^d.
$$

\begin{sidewaystable}[ht]
\begin{center}
\caption{Percentage of rejection of $H_0$ computed from 1000 random samples of size $n$ generated under $H_0$ defined in~(\ref{H0}), where $F_0$ is a bivariate c.d.f.\ whose univariate margins $F^{[1]}_0$ and $F^{[2]}_0$ are either both standard normal (N) or both standard exponential (E), and whose copula is either the Clayton (Cl) or the Gumbel--Hougaard (GH) with a Kendall's tau of $\tau$. The parameter $m$ used to uniformly discretize $\Smc_2^+$ was set to 8 (see Appendix~\ref{implementation} for more details).}
\label{H0biv}
\begin{tabular}{lrrrrrrrrrrrrrrrrrr}
  \hline
  \multicolumn{3}{c}{} & \multicolumn{8}{c}{Cl} & \multicolumn{8}{c}{GH} \\ \cmidrule(lr){4-11} \cmidrule(lr){12-19} $F^{[i]}_0$ & $n$ & $\tau$ & $\widecheck{S}_{n,\vee}$ & $\widecheck{S}_{n,+}$ & $\widecheck{T}_{n,\vee}$ & $\widecheck{T}_{n,+}$ & $\widecheck{U}_{n,\vee}$ & $\widecheck{U}_{n,+}$ & $\widecheck{V}_{n,\vee}$ & $\widecheck{V}_{n,+}$ & $\widecheck{S}_{n,\vee}$ & $\widecheck{S}_{n,+}$ & $\widecheck{T}_{n,\vee}$ & $\widecheck{T}_{n,+}$ & $\widecheck{U}_{n,\vee}$ & $\widecheck{U}_{n,+}$ & $\widecheck{V}_{n,\vee}$ & $\widecheck{V}_{n,+}$ \\ \hline
N & 50 & 0.00 & 3.8 & 3.7 & 4.7 & 4.8 & 5.5 & 4.4 & 5.9 & 5.9 & 4.2 & 3.8 & 5.4 & 4.0 & 4.7 & 3.5 & 5.8 & 5.7 \\ 
  N & 50 & 0.25 & 3.2 & 2.6 & 4.8 & 4.0 & 5.5 & 4.1 & 6.3 & 5.7 & 4.9 & 4.4 & 4.5 & 4.2 & 5.2 & 4.4 & 6.2 & 5.4 \\ 
  N & 50 & 0.50 & 5.3 & 4.9 & 6.3 & 6.2 & 4.9 & 4.4 & 5.1 & 5.4 & 4.7 & 4.4 & 5.8 & 5.3 & 4.9 & 4.3 & 5.3 & 5.2 \\ 
  N & 50 & 0.75 & 5.0 & 4.7 & 5.0 & 4.6 & 4.9 & 4.5 & 6.4 & 5.7 & 4.4 & 3.9 & 4.5 & 4.1 & 4.2 & 3.8 & 4.6 & 4.5 \\ 
  N & 100 & 0.00 & 4.9 & 4.9 & 5.5 & 6.2 & 5.6 & 4.7 & 6.5 & 6.6 & 4.0 & 4.5 & 5.2 & 4.7 & 3.5 & 3.4 & 4.0 & 4.2 \\ 
  N & 100 & 0.25 & 5.4 & 4.9 & 5.6 & 5.1 & 4.7 & 4.3 & 5.6 & 5.6 & 6.0 & 4.9 & 5.9 & 6.3 & 5.3 & 4.4 & 6.5 & 6.1 \\ 
  N & 100 & 0.50 & 5.1 & 4.2 & 5.1 & 4.9 & 4.1 & 4.5 & 6.0 & 6.3 & 4.7 & 4.2 & 4.3 & 3.9 & 3.8 & 4.2 & 5.7 & 5.4 \\ 
  N & 100 & 0.75 & 4.9 & 5.4 & 5.5 & 5.7 & 5.1 & 5.3 & 6.5 & 7.4 & 5.1 & 4.7 & 4.5 & 3.6 & 5.5 & 4.2 & 5.3 & 5.1 \\ 
  N & 200 & 0.00 & 5.5 & 5.4 & 5.5 & 6.0 & 5.1 & 6.0 & 5.7 & 7.4 & 4.8 & 4.2 & 5.6 & 4.6 & 4.0 & 3.7 & 5.3 & 4.7 \\ 
  N & 200 & 0.25 & 5.4 & 5.2 & 5.9 & 6.4 & 4.9 & 4.4 & 6.0 & 6.3 & 5.2 & 4.1 & 5.3 & 5.2 & 4.9 & 4.3 & 5.7 & 4.8 \\ 
  N & 200 & 0.50 & 5.8 & 4.8 & 4.9 & 4.1 & 5.2 & 4.4 & 4.9 & 4.9 & 5.1 & 6.0 & 5.5 & 6.0 & 6.0 & 6.0 & 5.8 & 7.1 \\ 
  N & 200 & 0.75 & 6.7 & 5.5 & 7.3 & 6.0 & 6.7 & 6.1 & 7.1 & 6.3 & 4.9 & 4.5 & 5.4 & 4.7 & 5.8 & 4.6 & 5.6 & 5.6 \\ 
  E & 50 & 0.00 & 4.3 & 4.8 & 6.0 & 6.0 & 4.0 & 3.9 & 6.1 & 6.1 & 3.8 & 3.9 & 3.8 & 3.6 & 3.6 & 3.2 & 5.8 & 5.7 \\ 
  E & 50 & 0.25 & 3.7 & 3.4 & 5.3 & 4.0 & 4.7 & 3.6 & 5.8 & 4.4 & 6.6 & 5.6 & 6.9 & 5.9 & 5.2 & 4.2 & 6.7 & 5.8 \\ 
  E & 50 & 0.50 & 3.7 & 3.4 & 4.8 & 5.0 & 3.3 & 3.4 & 6.5 & 6.2 & 4.8 & 4.7 & 5.9 & 5.1 & 4.7 & 3.6 & 6.9 & 5.4 \\ 
  E & 50 & 0.75 & 5.4 & 4.9 & 6.0 & 5.4 & 5.2 & 4.2 & 5.9 & 6.4 & 5.9 & 5.2 & 6.0 & 5.2 & 6.0 & 5.2 & 7.6 & 6.2 \\ 
  E & 100 & 0.00 & 4.9 & 4.6 & 5.3 & 4.6 & 4.6 & 4.8 & 4.7 & 5.1 & 5.3 & 4.5 & 5.2 & 4.6 & 5.5 & 4.6 & 5.5 & 5.7 \\ 
  E & 100 & 0.25 & 5.1 & 4.9 & 4.9 & 4.7 & 5.9 & 4.8 & 5.7 & 6.1 & 3.5 & 3.6 & 4.4 & 4.6 & 4.1 & 4.6 & 5.6 & 6.5 \\ 
  E & 100 & 0.50 & 5.4 & 4.6 & 5.3 & 4.5 & 5.3 & 4.4 & 6.3 & 5.1 & 5.2 & 4.4 & 4.0 & 3.8 & 4.9 & 3.8 & 5.5 & 5.9 \\ 
  E & 100 & 0.75 & 5.5 & 5.8 & 5.5 & 5.6 & 5.1 & 5.0 & 5.3 & 5.6 & 6.1 & 6.5 & 7.2 & 7.1 & 6.6 & 6.0 & 6.9 & 6.8 \\ 
  E & 200 & 0.00 & 4.9 & 4.4 & 5.0 & 4.9 & 5.2 & 4.6 & 5.1 & 5.4 & 5.1 & 4.4 & 4.6 & 4.9 & 4.5 & 4.7 & 5.5 & 5.6 \\ 
  E & 200 & 0.25 & 5.6 & 5.6 & 5.4 & 5.3 & 5.6 & 5.3 & 6.6 & 6.4 & 4.5 & 4.6 & 5.5 & 5.0 & 4.6 & 4.9 & 4.7 & 5.7 \\ 
  E & 200 & 0.50 & 5.5 & 4.8 & 6.4 & 5.3 & 5.1 & 4.8 & 5.8 & 5.5 & 6.7 & 6.2 & 6.3 & 6.0 & 6.1 & 6.2 & 5.6 & 6.5 \\ 
  E & 200 & 0.75 & 4.9 & 4.3 & 4.9 & 5.0 & 4.7 & 4.6 & 6.2 & 6.3 & 6.8 & 5.6 & 5.8 & 5.5 & 6.7 & 5.9 & 7.4 & 6.7 \\ 
   \hline
\end{tabular}
\end{center}
\end{sidewaystable}

Table~\ref{H0biv} reports rejection percentages under $H_0$ when $F_0$ is a bivariate c.d.f.\ with Clayton or Gumbel--Hougaard copula and standard exponential or standard normal margins. The parameter of the copula is chosen so that the theoretical value of Kendall's tau is equal to $\tau \in \{0,0.25,0.5,0.75\}$. Note that $\tau=0$ corresponds to independence while $\tau=0.5$ entails moderate dependence between the two components. Note also that the parameter $m$ used in to discretize uniformly $\Smc_2^+$ was set to 8 (see Appendix~\ref{implementation} for more details). The settings $m=6$ and $m=10$ were also considered but this did not seem to affect the results much. As the tests whose approximate $p$-value is computed using the processes $\widehat{\D}_n^{(j)}$ appeared systematically too liberal for the sample sizes under consideration, we only report the results of the tests based on the $\widecheck{\D}_n^{(j)}$ in Table~\ref{H0biv}. As can be seen, the latter tests seem to hold their level reasonably well except perhaps the tests based on $\widecheck{V}_{n,\vee}$ and $\widecheck{V}_{n,+}$ which might be slightly too liberal for $\tau = 0.75$ and the sample sizes under consideration. Similar results were obtained in dimension three. 

\begin{sidewaystable}[ht]
\begin{center}
\caption{Percentage of rejection of $H_0$ computed from 1000 samples of size $n$ generated under $H_1$ defined in~(\ref{H1}),  where $k^\star = \ip{nt}$, $F_1$ and $F_2$ are bivariate c.d.f.s that only differ in their first margin which is standard exponential for $F_1$ and exponential with rate 0.5 for $F_2$. The common copula $C$ of $F_1$ and $F_2$ is either the Clayton (Cl) or the Gumbel--Hougaard (GH) with a Kendall's tau of $\tau$. The second margin of both $F_1$ and $F_2$ is standard exponential. The parameter $m$ used to uniformly discretize $\Smc_2^+$ was set to 8.}
\label{H1marg}
\begin{tabular}{rrrrrrrrrrrrrrrrrrr}
  \hline
  \multicolumn{3}{c}{} & \multicolumn{8}{c}{Cl} & \multicolumn{8}{c}{GH} \\ \cmidrule(lr){4-11} \cmidrule(lr){12-19} $n$ & $\tau$ & $t$ &$\widecheck{S}_{n,\vee}$ & $\widecheck{S}_{n,+}$ & $\widecheck{T}_{n,\vee}$ & $\widecheck{T}_{n,+}$ & $\widecheck{U}_{n,\vee}$ & $\widecheck{U}_{n,+}$ & $\widecheck{V}_{n,\vee}$ & $\widecheck{V}_{n,+}$ & $\widecheck{S}_{n,\vee}$ & $\widecheck{S}_{n,+}$ & $\widecheck{T}_{n,\vee}$ & $\widecheck{T}_{n,+}$ & $\widecheck{U}_{n,\vee}$ & $\widecheck{U}_{n,+}$ & $\widecheck{V}_{n,\vee}$ & $\widecheck{V}_{n,+}$ \\ \hline
50 & 0.0 & 0.10 & 5.3 & 5.1 & 5.7 & 5.4 & 6.1 & 6.8 & 7.4 & 9.0 & 5.6 & 5.6 & 6.5 & 7.2 & 6.8 & 7.2 & 7.9 & 9.0 \\ 
  50 & 0.0 & 0.25 & 12.4 & 12.4 & 12.3 & 13.8 & 19.6 & 19.1 & 16.6 & 18.4 & 10.6 & 11.0 & 12.8 & 13.4 & 19.4 & 19.5 & 17.4 & 19.1 \\ 
  50 & 0.0 & 0.50 & 21.3 & 19.9 & 23.8 & 22.2 & 41.2 & 36.7 & 35.0 & 35.2 & 24.6 & 23.4 & 25.2 & 23.5 & 43.4 & 39.5 & 35.2 & 35.4 \\ 
  50 & 0.5 & 0.10 & 4.3 & 3.9 & 4.8 & 4.8 & 4.3 & 4.8 & 6.4 & 8.4 & 5.0 & 5.2 & 5.9 & 5.5 & 7.3 & 7.7 & 9.3 & 10.6 \\ 
  50 & 0.5 & 0.25 & 8.3 & 7.9 & 11.6 & 12.1 & 19.5 & 20.8 & 26.3 & 30.2 & 8.7 & 8.7 & 11.8 & 11.8 & 21.6 & 25.2 & 30.0 & 32.9 \\ 
  50 & 0.5 & 0.50 & 17.5 & 15.3 & 26.1 & 22.9 & 49.2 & 50.1 & 59.5 & 56.8 & 18.6 & 18.2 & 26.9 & 23.9 & 48.5 & 48.7 & 63.1 & 58.1 \\ 
  100 & 0.0 & 0.10 & 7.9 & 8.3 & 7.9 & 9.7 & 9.8 & 10.4 & 9.9 & 11.9 & 7.3 & 7.0 & 8.7 & 9.3 & 10.1 & 10.6 & 9.1 & 11.6 \\ 
  100 & 0.0 & 0.25 & 25.0 & 23.9 & 23.7 & 26.1 & 44.5 & 42.7 & 33.5 & 35.9 & 25.9 & 25.9 & 27.5 & 30.6 & 47.3 & 45.3 & 34.6 & 40.0 \\ 
  100 & 0.0 & 0.50 & 44.1 & 43.1 & 49.9 & 47.7 & 72.1 & 68.6 & 63.8 & 61.6 & 45.8 & 44.6 & 53.5 & 49.8 & 74.2 & 70.0 & 63.8 & 62.3 \\ 
  100 & 0.5 & 0.10 & 5.9 & 5.5 & 6.4 & 7.3 & 8.9 & 9.7 & 11.6 & 16.4 & 6.6 & 6.6 & 6.9 & 8.4 & 9.3 & 11.1 & 10.8 & 18.4 \\ 
  100 & 0.5 & 0.25 & 14.7 & 15.6 & 22.4 & 24.5 & 50.8 & 53.7 & 57.5 & 60.3 & 17.9 & 17.3 & 28.1 & 28.0 & 59.8 & 61.6 & 70.1 & 70.5 \\ 
  100 & 0.5 & 0.50 & 28.2 & 30.2 & 49.0 & 48.7 & 84.7 & 85.8 & 89.2 & 87.6 & 34.2 & 34.3 & 51.4 & 48.6 & 87.7 & 88.8 & 94.2 & 92.6 \\ 
  200 & 0.0 & 0.10 & 11.2 & 11.5 & 10.5 & 14.4 & 15.7 & 19.1 & 11.6 & 19.3 & 11.2 & 12.0 & 10.3 & 13.9 & 17.4 & 19.7 & 12.8 & 19.4 \\ 
  200 & 0.0 & 0.25 & 49.6 & 48.4 & 56.9 & 57.1 & 81.2 & 77.5 & 67.5 & 68.6 & 52.6 & 51.5 & 58.6 & 58.4 & 81.4 & 79.0 & 67.0 & 71.0 \\ 
  200 & 0.0 & 0.50 & 80.1 & 78.5 & 87.4 & 84.9 & 97.0 & 95.8 & 92.6 & 91.6 & 77.3 & 75.9 & 85.9 & 81.9 & 97.2 & 96.2 & 92.1 & 90.4 \\ 
  200 & 0.5 & 0.10 & 7.0 & 7.8 & 9.9 & 14.6 & 16.3 & 22.9 & 21.1 & 30.2 & 8.0 & 8.4 & 8.7 & 12.3 & 18.6 & 24.3 & 22.2 & 35.2 \\ 
  200 & 0.5 & 0.25 & 25.8 & 28.3 & 53.6 & 54.2 & 87.7 & 88.7 & 90.1 & 90.4 & 30.4 & 32.8 & 55.6 & 54.7 & 95.0 & 94.6 & 97.7 & 95.9 \\ 
  200 & 0.5 & 0.50 & 52.9 & 57.0 & 83.9 & 82.3 & 99.4 & 99.3 & 99.4 & 99.5 & 67.5 & 66.9 & 86.7 & 84.7 & 99.9 & 99.8 & 100.0 & 100.0 \\ 
   \hline
\end{tabular}
\end{center}
\end{sidewaystable}

Estimated rejection rates when both distributions in $H_1$ have the same copula (Clayton or Gumbel--Hougaard) but differ in one margin are given in Table~\ref{H1marg}. To be precise, the bivariate distributions $F_1$ and $F_2$ only differ in the first margin which is standard exponential for $F_1$ and exponential with rate 0.5 for $F_2$, while the second margin of both $F_1$ and $F_2$ is standard exponential. As one can see, the tests based on half-spaces are substantially more powerful than the tests based on lower-left orthants. The test based on $\widecheck{U}_{n,\vee}$ is the most powerful, overall, when $\tau = 0$, while, when $\tau=0.5$, it is the test based on $\widecheck{V}_{n,+}$ that has, overall, the highest rejection rates.

\begin{sidewaystable}[ht]
\begin{center}
\caption{Percentage of rejection of $H_0$ computed from 1000 samples of size $n$ generated under $H_1$ defined in~(\ref{H1}), where $k^\star = \ip{nt}$, $F_1$ and $F_2$ are bivariate c.d.f.s with standard exponential margins and copula from the same family (Clayton (Cl) or Gumbel--Hougaard (GH)) but with different parameter values. The copula of $F_1$ has a Kendall's tau of 0.1, while that of $F_2$ has a Kendall's tau of $\tau$. The parameter $m$ used to uniformly discretize $\Smc_2^+$ was set to 8.}
\label{H1cop}
\begin{tabular}{rrrrrrrrrrrrrrrrrrr}
  \hline
  \multicolumn{3}{c}{} & \multicolumn{8}{c}{Cl} & \multicolumn{8}{c}{GH} \\ \cmidrule(lr){4-11} \cmidrule(lr){12-19} $n$ & $\tau$ & $t$ & $\widecheck{S}_{n,\vee}$ & $\widecheck{S}_{n,+}$ & $\widecheck{T}_{n,\vee}$ & $\widecheck{T}_{n,+}$ & $\widecheck{U}_{n,\vee}$ & $\widecheck{U}_{n,+}$ & $\widecheck{V}_{n,\vee}$ & $\widecheck{V}_{n,+}$ & $\widecheck{S}_{n,\vee}$ & $\widecheck{S}_{n,+}$ & $\widecheck{T}_{n,\vee}$ & $\widecheck{T}_{n,+}$ & $\widecheck{U}_{n,\vee}$ & $\widecheck{U}_{n,+}$ & $\widecheck{V}_{n,\vee}$ & $\widecheck{V}_{n,+}$ \\ \hline
50 & 0.3 & 0.10 & 5.5 & 5.9 & 5.2 & 5.1 & 4.3 & 3.8 & 5.2 & 6.0 & 4.7 & 5.2 & 5.1 & 5.1 & 5.8 & 4.9 & 6.8 & 5.8 \\ 
  50 & 0.3 & 0.25 & 5.9 & 5.1 & 6.8 & 5.3 & 5.4 & 4.7 & 6.2 & 6.5 & 5.6 & 5.1 & 6.0 & 5.6 & 5.9 & 4.2 & 7.2 & 6.2 \\ 
  50 & 0.3 & 0.50 & 4.4 & 4.5 & 6.8 & 4.6 & 3.9 & 2.9 & 6.2 & 4.9 & 6.1 & 6.1 & 6.3 & 6.3 & 4.7 & 4.7 & 6.8 & 6.7 \\ 
  50 & 0.7 & 0.10 & 5.7 & 5.7 & 6.3 & 6.2 & 5.1 & 4.9 & 6.4 & 6.4 & 5.4 & 5.6 & 5.8 & 5.6 & 3.8 & 4.2 & 6.3 & 6.6 \\ 
  50 & 0.7 & 0.25 & 7.5 & 8.0 & 6.8 & 8.1 & 4.4 & 3.8 & 6.9 & 9.3 & 7.3 & 9.5 & 9.1 & 8.9 & 5.8 & 5.5 & 8.0 & 9.8 \\ 
  50 & 0.7 & 0.50 & 13.4 & 13.7 & 13.9 & 14.3 & 7.1 & 8.3 & 13.9 & 14.8 & 12.8 & 13.1 & 12.6 & 11.8 & 6.1 & 6.9 & 13.3 & 14.0 \\ 
  100 & 0.3 & 0.10 & 4.6 & 4.8 & 5.9 & 5.8 & 4.4 & 4.5 & 6.5 & 6.4 & 5.0 & 5.2 & 5.9 & 6.4 & 6.4 & 5.7 & 8.0 & 7.6 \\ 
  100 & 0.3 & 0.25 & 4.8 & 5.1 & 5.4 & 5.3 & 4.7 & 4.4 & 5.2 & 5.8 & 5.3 & 5.6 & 6.4 & 7.2 & 4.8 & 4.3 & 5.4 & 6.4 \\ 
  100 & 0.3 & 0.50 & 7.6 & 7.4 & 7.2 & 6.5 & 5.3 & 4.4 & 6.2 & 6.1 & 7.0 & 6.8 & 7.0 & 7.0 & 4.4 & 4.5 & 5.8 & 6.9 \\ 
  100 & 0.7 & 0.10 & 7.1 & 6.5 & 7.2 & 6.8 & 5.6 & 4.8 & 5.8 & 8.7 & 5.1 & 5.8 & 5.9 & 5.8 & 3.8 & 5.0 & 5.4 & 7.0 \\ 
  100 & 0.7 & 0.25 & 12.9 & 13.4 & 13.0 & 13.4 & 6.9 & 8.5 & 12.0 & 18.6 & 13.9 & 15.7 & 12.7 & 13.6 & 5.8 & 7.7 & 13.5 & 21.2 \\ 
  100 & 0.7 & 0.50 & 25.4 & 26.7 & 23.8 & 25.0 & 11.4 & 15.0 & 32.5 & 36.0 & 27.1 & 27.1 & 25.0 & 25.7 & 12.7 & 16.6 & 34.2 & 39.4 \\ 
  200 & 0.3 & 0.10 & 6.0 & 5.9 & 5.0 & 5.2 & 5.6 & 4.5 & 5.4 & 4.6 & 4.5 & 4.4 & 4.9 & 5.3 & 4.9 & 4.9 & 6.0 & 6.2 \\ 
  200 & 0.3 & 0.25 & 6.2 & 7.1 & 5.9 & 6.9 & 4.7 & 5.6 & 5.7 & 7.1 & 6.0 & 6.2 & 5.8 & 5.5 & 4.7 & 5.3 & 6.2 & 6.8 \\ 
  200 & 0.3 & 0.50 & 7.9 & 8.4 & 7.4 & 7.4 & 5.9 & 5.5 & 6.9 & 7.2 & 10.8 & 10.6 & 9.3 & 10.4 & 7.2 & 7.2 & 8.6 & 9.3 \\ 
  200 & 0.7 & 0.10 & 8.9 & 11.2 & 9.1 & 11.5 & 6.5 & 7.1 & 8.0 & 10.8 & 8.9 & 9.6 & 8.7 & 11.1 & 5.1 & 6.7 & 8.1 & 11.8 \\ 
  200 & 0.7 & 0.25 & 26.0 & 28.1 & 25.3 & 25.3 & 11.0 & 15.0 & 25.1 & 29.6 & 28.6 & 30.0 & 25.9 & 28.3 & 14.3 & 20.4 & 30.1 & 43.0 \\ 
  200 & 0.7 & 0.50 & 43.5 & 44.8 & 43.7 & 44.4 & 22.6 & 30.4 & 56.8 & 61.4 & 48.0 & 45.9 & 45.2 & 44.5 & 35.2 & 51.7 & 79.8 & 85.2 \\ 
   \hline
\end{tabular}
\end{center}
\end{sidewaystable}

Table~\ref{H1cop} reports rejection percentages when the change in distribution is only due to a change in the dependence structure: both $F_1$ and $F_2$ have standard exponential margins but the copula of $F_1$ is the Clayton (resp.\ Gumbel--Hougaard) copula with a Kendall's tau of 0.1, while that of $F_2$ is the Clayton (resp.\ Gumbel--Hougaard) copula with a Kendall's tau of $\tau \in \{0.3,0.7\}$. As one can see from the overall low rejection percentages, this problem appears to be more difficult than the previous one. The tests have hardly any power for $\tau = 0.3$. For $\tau = 0.7$, the tests based on $\widecheck{V}_{n,\vee}$ and $\widecheck{V}_{n,+}$ are the most powerful, overall, although one should be cautious as they might be slightly too liberal in the case of strongly dependent data according to Table~\ref{H0biv}.

The setting used to obtain Table~\ref{H1cop} was finally extended to dimension three (results not reported). The conclusions are very similar to those obtained in dimension two with the difference that all the rejection rates are higher. Hence, as could have been expected, detecting a change in the dependence becomes easier as the dimension increases.

\section{Practical recommendations and illustration}
\label{illustration}

From the results of the Monte Carlo experiments partially reported in the previous section, the tests based on $\widecheck{S}_{n,\vee}$ and $\widecheck{T}_{n,+}$ seem good choices in the univariate case, while the test based on $\widecheck{V}_{n,+}$ seems to be a good one in the multivariate case.

As an illustration, we applied the tests based on the processes $\widecheck{\D}_n^{(j)}$ to the trivariate hydrological data collected at the Ceppo Morelli dam, Italy, studied in \cite{SalDeMDur11}. The data consist of annual maxima for 49 years (in the range 1937--1994) of three variables: L (dam reservoir water level in $m$), Q (peak flow in $m^3.s^{-1}$) and V (peak volume in $10^6$ $m^3$). For a detailed description of the data, we refer the reader to Section~2 of \cite{SalDeMDur11}.

\begin{table}[ht]
\begin{center}
\caption{Approximate $p$-values of the tests based on the processes $\widecheck{\D}_n^{(j)}$ for the trivariate hydrological data considered in Section~\ref{illustration}. The trivariate (resp.\ bivariate) tests based on half-spaces were run with $m=32$ (resp.\ $m=8$). The approximate $p$-values were computed from $N=10,000$ multiplier realizations.}
\label{pval}
\begin{tabular}{crrrrrrrr}
  \hline
  Variables & $\widecheck{S}_{n,\vee}$ & $\widecheck{S}_{n,+}$ & $\widecheck{T}_{n,\vee}$ & $\widecheck{T}_{n,+}$ & $\widecheck{U}_{n,\vee}$ & $\widecheck{U}_{n,+}$ & $\widecheck{V}_{n,\vee}$ & $\widecheck{V}_{n,+}$ \\ \hline
(L,Q,V) & 0.114 & 0.120 & 0.015 & 0.028 & 0.010 & 0.010 & 0.004 & 0.006 \\ 
  L & 0.479 & 0.314 & 0.510 & 0.236 &  &  &  &  \\ 
  (Q,V) & 0.024 & 0.028 & 0.012 & 0.012 & 0.015 & 0.014 & 0.004 & 0.007 \\ 
   \hline
\end{tabular}
\end{center}
\end{table}

We first tested for a change in the distribution of (L,Q,V). The approximate $p$-values of the eight tests based on the processes $\widecheck{\D}_n^{(j)}$ are given in the first line of Table~\ref{pval}. Since there are both physical and statistical reasons to believe that L is independent of (Q,V) as explained in \cite{SalDeMDur11}, as a second step, we tested for a change in the distribution of L and in the distribution of (Q,V) separately. The obtained approximate $p$-values are reported in the second and third lines of Table~\ref{pval}. As can be seen from the results of the test based on $\widecheck{V}_{n,+}$, there is strong evidence of a change in the distributions of (L,Q,V) and (Q,V). From the second line of Table~\ref{pval}, we see that, on the contrary, there is no evidence of a change in the distribution of L. The latter finding is completely consistent with the fact that the variability of L is mainly due to the management policy of the reservoir which is constant over time. Indeed, as explained in \cite{SalDeMDur11}, the target of the dam manager is to keep a high water level in order to maximize electricity production. 

As classically done in the literature, under the hypothesis of a single break in a distribution, the change-point can be estimated by one of $\arg \max_{1 \leq k \leq n-1} \widecheck{S}_{n,k}$, $\arg \max_{1 \leq k \leq n-1} \widecheck{T}_{n,k}$, $\arg \max_{1 \leq k \leq n-1} \widecheck{U}_{n,k}$ or $\arg \max_{1 \leq k \leq n-1} \widecheck{V}_{n,k}$ depending on which test one wants to consider. For instance, the last estimator gives 31 for both $(L,Q,V)$ and $(Q,V)$, which corresponds to a change after the year 1976. 

Finally, let us mention that the approach based on multivariate empirical c.d.f.s considered in \citet[Section 2.6]{CsoHor97} and in \cite{GomHor99} has been extended by \cite{Ino01} to serially dependent observations, although the latter work is not aware of the former ones. A future research direction would be to study generalizations as such proposed in this work in the setting considered by \cite{Ino01}. 

\section*{Acknowledgments}

The authors would like to thank two anonymous referees for their very insightful and constructive comments which helped to improve the paper. The authors are also very grateful to Johan Segers for his help with Lemma~\ref{phi}, and to Gianfausto Salvadori for providing the trivariate hydrological data analyzed in the last section. This work was started in September 2011 during Mark Holmes's visit to University of Pau.  We gratefully acknowledge the support received during this visit.


\appendix

\section{Proofs}

\subsection{Proof of Theorem \ref{multiplier}}

To prove Theorem \ref{multiplier}, we first show a lemma.

\begin{lem}
\label{tildeZn_hatZn}
Let $\FF$ be a $P$-Donsker class of functions with measurable envelope $F_e$ such that $P F_e^2 < \infty$. Then,
$$
\sup_{s \in [0,1]} \| \widetilde{\Z}_n(s,f) - \widehat{\Z}_n(s,f) \|_\FF \aso 0.
$$
\end{lem}
\begin{proof}
For any $s \in [0,1]$ and $f \in \FF$, we have
$$
\widetilde{\Z}_n(s,f) - \widehat{\Z}_n(s,f) = \left( \frac{1}{\sqrt{n}} \sum_{i=1}^{\ip{ns}} \xi_i \right) ( \P_{\ip{ns}} f - Pf).
$$

Now, let $Y_k=\| \P_{k} f - Pf\|_{\FF}$, $k \in \{1,\dots,n\}$.  Note that $Y_k$ need not be measurable, but, by the assumption that $\sup_{f\in \FF}|f(x)|\le F_e(x)$ for all $x \in \R^d$, we have that 
$$
\| \P_{k} f - Pf\|_{\FF}\le \P_kF_e+PF_e<\infty,
$$
where $\P_kF_e+PF_e$ is measurable. Thus, for each $k$, there exists a smallest random variable $Y_k^* < \infty$ such that $Y_k(\omega)\le Y_k^*(\omega)$ for every $\omega$ \citep[see][Lemma 1.2.1]{vanWel96}.

To prove the claim, it suffices to show that 
$$
A_n = \max_{3\le k\le n} Y_k^*  \times \frac{1}{\sqrt{n}} \sum_{i=1}^k \xi_i \as 0,
$$
where the variables are all measurable. Now, for any $n \geq 3$, let $a_n = n^{-1/2} (\log \log n)^{1/2}$. Then,
$$
A_n = \max_{3\le k\le n} \frac{Y_k^*}{a_k}  \times \frac{a_k^2 k}{n^{1/2}} \times \frac{\bar \xi_k}{a_k} \leq \max_{3\le k\le n} \frac{Y_k^*}{a_k} \times n^{-1/2} \log\log n \times \max_{3\le k\le n} \frac{\bar \xi_k}{a_k},
$$
where $\bar \xi_k = k^{-1} \sum_{i=1}^k \xi_i$. 

Since $\FF$ is $P$-Donsker with measurable envelope function $F_e$ satisfying $P F_e^2 < \infty$, we know from the law of the iterated logarithm for empirical processes \citep[page 31]{DudPhi83,Kos08} that $\limsup Y^*_n/a_n < \infty$ almost surely, which implies that $\max_{3\le k\le n} Y_k^* / a_k \leq \sup_{k \geq 3} Y_k^* / a_k < \infty$ almost surely. Similarly, from the law of the iterated logarithm for the mean of an i.i.d.\ sequence with expectation~0 and variance~1, we have $\max_{3\le k\le n} \bar \xi_k/a_k \leq \sup_{k \geq 3} \bar \xi_k/a_k < \infty$ almost surely. The desired result finally follows from the fact that $n^{-1/2} \log\log n \to 0$.
\end{proof}

\begin{proof}[\bf Proof of Theorem \ref{multiplier}.] 
The proof of the weak convergence of the finite-dimensional marginal distributions of $(\Z_n,\widetilde{\Z}_n)$ to those of $(\Z_P,\Z'_P)$ is a more complicated version of the corresponding result for convergence of the rescaled random walk increments to Brownian motion. It is omitted here for the sake of brevity.

To obtain that $(\Z_n,\widetilde{\Z}_n ) \leadsto (\Z_P,\Z'_P)$ in $\{ \ell^\infty([0,1] \times \FF) \}^2$, it remains to show that $(\Z_n,\widetilde{\Z}_n)$ is asymptotically tight \citep[see e.g.][Section 18.3]{van98}, which holds if both $\Z_n$ and $\widetilde{\Z}_n$ converge weakly to tight random elements. 

From Theorem 2.12.1 of \cite{vanWel96}, we have that $\Z_n \leadsto \Z_P$ in $\ell^\infty([0,1] \times \FF)$. Now, let
$$
\G_n' = \frac{1}{\sqrt{n}} \sum_{i=1}^n \xi_i (\delta_{X_i} - P) 
$$
be a {\em multiplier} version of $\G_n$. The class $\FF$ being $P$-Donsker, from the functional unconditional multiplier central limit theorem \citep[see e.g.][]{vanWel96,Kos08}, we have that $\G_n' \leadsto \G'_P$ in $\ell^\infty(\FF)$, where $\G'_P$ is an independent copy of the $P$-Brownian bridge $\G_P$, which implies that $\G_n'$ is asymptotically tight. To show that $\widetilde{\Z}_n \leadsto \Z'_P$, where $\Z'_P$ is an independent copy of $\Z_P$, one can use the asymptotic tightness of $\G_n'$ and proceed as in the proof of Theorem 2.12.1 of \cite{vanWel96}. Note that the proof can be further simplified if the process $\widetilde{\Z}_n$ is symmetric in the sense of Chapter A.1 of \cite{vanWel96}, which happens if $\xi_1,\dots,\xi_n$ are symmetrically distributed around zero. In that case, Ottaviani's inequality can be replaced by L\'evy's inequality \citep[see e.g.][Proposition A.1.2]{vanWel96}, which shortens the argument.

Hence, we have that $(\Z_n,\widetilde{\Z}_n ) \leadsto (\Z_P,\Z'_P )$ in $\{ \ell^\infty([0,1] \times \FF) \}^2$. The fact that $\Z'_P$ is independent of $\Z_P$ comes via that the finite-dimensional distributions, which are multivariate normal, and the fact that $\Z_n$ and $\widetilde{\Z}_n$ are uncorrelated. 

Next, notice that 
$$
\sup_{s \in [0,1]} \| \widetilde{\Z}_n(s,f) - \widecheck{\Z}_n(s,f) \|_\FF  = \| \P_n f - Pf \|_\FF \times \sup_{s \in [0,1]} \left| \frac{1}{\sqrt{n}} \sum_{i=1}^{\ip{ns}} \xi_i \right| 
$$
converges in outer probability to zero because $\sup_{s \in [0,1]} | n^{-1/2} \sum_{i=1}^{\ip{ns}} \xi_i |$ converges weakly to the supremum of the absolute value of Brownian motion and $\| \P_n f - Pf \|_\FF \p 0$.

From the continuous mapping theorem, we have that $(\Z_n,\widetilde{\Z}_n,\widetilde{\Z}_n,\widetilde{\Z}_n ) \leadsto (\Z_P,\Z'_P,\Z'_P,\Z'_P )$ in $\{ \ell^\infty([0,1] \times \FF) \}^4$. The desired result finally follows from the previous remark and Lemma~\ref{tildeZn_hatZn}.
\end{proof}



\subsection{Proof of Theorem \ref{thm_H0}}

\begin{proof}
For any $s \in [0,1]$ and $f \in \FF$, let
$$
\widetilde{\D}_n(s,f) = \{ 1 - \lambda_n(s) \} \widetilde{\Z}_n(s,f) - \lambda_n(s) \{ \widetilde{\Z}_n(1,f) - \widetilde{\Z}_n(s,f) \}= \widetilde{\Z}_n(s,f) - \lambda_n(s) \widetilde{\Z}_n(1,f). 
$$
Then, from Theorem~\ref{multiplier},~(\ref{DnH0}), the definitions of $\widetilde{\D}_n$ and $\widecheck{\D}_n$, and the continuous mapping theorem, we obtain that 
$$
\left( \D_n,\widetilde{\D}_n,\widecheck{\D}_n \right) \leadsto \left( \D_{P_0},\D'_{P_0},\D'_{P_0} \right)
$$
in $\{ \ell^\infty([0,1] \times \FF) \}^3$. To obtain the desired result, it remains to show that $\widetilde{\D}_n - \widehat{\D}_n \p 0$ in $\ell^\infty([0,1] \times \FF)$. From the definitions of $\widetilde{\D}_n$ and $\widehat{\D}_n$, we have that 
\begin{multline*}
\sup_{s \in [0,1]} \left\| \widetilde{\D}_n(s,f) - \widehat{\D}_n(s,f) \right\|_\FF \leq \sup_{s \in [0,1]} \{ 1 - \lambda_n(s) \} \sup_{s \in [0,1]} \left\| \widetilde{\Z}_n(s,f) - \widehat{\Z}_n(s,f) \right\|_\FF \\ +  \sup_{s \in [0,1]} \lambda_n(s) \sup_{s \in [0,1]} \left\| \left\{ \widetilde{\Z}_n(1,f) - \widetilde{\Z}_n(s,f) \right\} -  \widehat{\Z}_n^\star(s,f)\right\|_\FF,
\end{multline*}
where $\widehat{\Z}_n^\star$ is defined in~(\ref{hatZnstar}). The second supremum over $s$ on the right of the previous inequality converges outer almost surely to zero according to Lemma~\ref{tildeZn_hatZn}. Furthermore, it can be verified that the last supremum over $s$ (written for instance as a maximum over $1 \leq \ip{ns} \leq n$) is nothing else than the version of the second supremum computed from the ``reversed'' sequence $(\xi_n,X_n),(\xi_{n-1},X_{n-1}),\dots,(\xi_1,X_1)$. As the latter sequence and the original sequence have the same distribution, Lemma~\ref{tildeZn_hatZn} implies that the last supremum converges in outer probability to zero.
\end{proof}


\subsection{Proof of Theorem \ref{thm_H1}}

To prove Theorem \ref{thm_H1}, we first prove a lemma.

\begin{lem}
\label{lem:boundprob_1}
Let $\xi_1,\ldots,\xi_n$ be i.i.d.\ random variables with mean 0, variance 1 and satisfying $\int_0^\infty \{ \Pr(|\xi_1| > x) \}^{1/2} \dd x < \infty$, let $X_1,\ldots,X_n$ be i.i.d.\ random variables with law $P$ independent of $\xi_1,\ldots,\xi_n$, and let $\FF$ be a $P$-Donsker class with measurable envelope function $F_e$ such that $PF_e<\infty$. Then, for every $0 \leq t_1 \leq t_2 \leq 1$ and every $\vep > 0$, there exists $M > 0$ such that
$$
\sup_{n \in \N} \Pr^* \left( \sup_{s \in [t_1,t_2] } \left\| \frac{1}{\sqrt{n}} \sum_{i=1}^{\ip{ns}} \xi_i f(X_i) \right\|_{\FF} > M \right) < \vep.
$$
\end{lem}

\begin{proof} 
Using the triangle inequality and the envelope $F_e$, for any $s \in [0,1]$ and $f \in \FF$, we have 
\begin{align*}
\left\| \frac{1}{\sqrt{n}} \sum_{i=1}^{\ip{ns}} \xi_i f(X_i) \right\|_{\FF} \le & \left\| \frac{1}{\sqrt{n}} \sum_{i=1}^{\ip{ns}} \xi_i \{ f(X_i)- Pf \} \right\|_{\FF} + \left\| \frac{1}{\sqrt{n}} \sum_{i=1}^{\ip{ns}} \xi_i Pf \right\|_{\FF} \\
\le &\left\| \widetilde{\Z}_n(s,f) \right\|_{\FF}+ PF_e \times \left| \frac{1}{\sqrt{n}} \sum_{i=1}^{\ip{ns}} \xi_i \right|. 
\end{align*}

Let $0 \leq t_1 \leq t_2 \leq 1$. Since $\widetilde{\Z}_n \leadsto \Z_P$ in $\ell^\infty([0,1] \times \FF)$, we have that $\sup_{s \in [t_1,t_2]} \| \widetilde{\Z}_n(s,f)\|_{\FF}$ is asymptotically tight. Let $\vep > 0$. Then, there exists $M' > 0$ such that
$$
\limsup_{n \to \infty} \Pr^* \left( \sup_{s \in [t_1,t_2] } \left\| \widetilde{\Z}_n(s,f) \right\|_{\FF} > M' \right) < \vep/4.
$$
It follows that, for $n \geq n_\vep$, 
$$
\sup_{n \geq n_\vep} \Pr^* \left( \sup_{s \in [t_1,t_2] } \left\| \widetilde{\Z}_n(s,f) \right\|_{\FF} > M' \right) < \vep/2.
$$
Since $\sup_{s \in [t_1,t_2] } \| \widetilde{\Z}_n(s,f) \|_{\FF}$ is bounded by an almost surely finite random variable for any $n < n_\vep$, there exists $M'' \geq M'$ such that 
$$
\sup_{n \in \N} \Pr^* \left( \sup_{s \in [t_1,t_2] } \left\| \widetilde{\Z}_n(s,f) \right\|_{\FF} > M'' \right) < \vep/2.
$$

Similarly, $s \mapsto n^{-1/2} \sum_{i=1}^{\ip{ns}} \xi_i$ converges weakly to Brownian motion in $\ell^\infty([0,1])$, and therefore $\sup_{s \in [t_1,t_2]} | n^{-1/2} \sum_{i=1}^{\ip{ns}} \xi_i |$ is uniformly tight. Proceeding as above, there exists $M''' > 0$ such that
$$
\sup_{n \in \N} \Pr \left( \sup_{s \in [t_1,t_2] }\left| \frac{1}{\sqrt{n}} \sum_{i=1}^{\ip{ns}} \xi_i \right| > M''' \right) < \vep/2.
$$
We get the desired result with $M = M''+ PF_e \times M'''$. 
\end{proof}


\begin{proof}[\bf Proof of Theorem \ref{thm_H1}~(i).] 
First, notice that, for any $s \in [0,1]$ and any $f \in \FF$, 
$$
K_t(s,f) =\left\{
\begin{array}{ll}
s P_1 f - s \{ t P_1 f + (1-t) P_2 f \} & \mbox{if } s \leq t, \\
t P_1 f + (s - t) P_2 f - s \{ t P_1 f + (1-t) P_2 f \} & \mbox{if } s > t. \\
\end{array}
\right.
$$
Let $\K_{n,t}(s) = \| \lambda_n(s) \left\{ 1 - \lambda_n(s) \right\} (
\P_{\ip{ns}}f - \P_{n-\ip{ns}}^\star f  ) - K_t(s,f) \|_\FF$. Clearly,
$$
\sup_{s \in [0,1]} \K_{n,t}(s) = \max \left\{ \sup_{s \in [0,t]} \K_{n,t}(s),\sup_{s \in (t,1]} \K_{n,t}(s) \right\}.
$$
Furthermore, for any $s \in [0,1]$ and $f \in \FF$, 
$$
\lambda_n(s) \left\{ 1 - \lambda_n(s) \right\} \left(
\P_{\ip{ns}}f - \P_{n-\ip{ns}}^\star f  \right) = \lambda_n(s)  \left(\P_{\ip{ns}}f - \P_n f \right).
$$
Hence, 
\begin{equation}
\label{decomp}
\sup_{s \in [0,t]} \K_{n,t}(s) \leq \sup_{s \in [0,t]} \left\| \lambda_n(s) \P_{\ip{ns}}f - s P_1 f \right\|_\FF + \sup_{s \in [0,1]} \left\| \lambda_n(s) \P_n f  - s \{ t P_1 f + (1-t) P_2 f \} \right\|_\FF.
\end{equation}
The first term on the right of the previous inequality is smaller than
\begin{multline*}
\sup_{s \in [0,t]} \left\| \lambda_n(s) \left( \P_{\ip{ns}}f - P_1 f \right) \right\|_\FF + \sup_{s \in [0,t]} \left\| \{ \lambda_n(s) - s \} P_1 f \right\|_\FF \\ = \frac{1}{\sqrt{n}} \sup_{s \in [0,t]} \left\| \Z_n(s,f) \right\|_\FF + \sup_{s \in [0,t]}  |\lambda_n(s) - s | \times \left\| P_1 f \right\|_\FF
\end{multline*}
and therefore converges in outer probability to zero because $\sup_{s \in [0,t]} \left\| \Z_n(s,f) \right\|_\FF$ converges in distribution. The second supremum on the right of~(\ref{decomp}) is smaller than  
$$
\sup_{s \in [0,1]} | \lambda_n(s) - s | \times \left\| \P_n f \right\|_\FF + \left\| \P_n f  - \{ t P_1 f + (1-t) P_2 f \} \right\|_\FF
$$
and therefore converges in outer probability to zero because, $\FF$ being $P_1$ and $P_2$-Donsker, $\P_n f$ converges in outer probability to $t P_1 f + (1-t) P_2 f$ uniformly in $f \in \FF$.

Similarly,
\begin{multline*}
\sup_{s \in (t,1]} \K_{n,t}(s) \leq \sup_{s \in (t,1]} \left\| \lambda_n(s) \P_{\ip{ns}}f - t P_1 f - (s - t) P_2 f \right\|_\FF \\ + \sup_{s \in [0,1]} \left\| \lambda_n(s) \P_n f  - s \{ t P_1 f + (1-t) P_2 f \} \right\|_\FF.
\end{multline*}
We already know that the second supremum on the right converges to zero in outer probability. Using the fact that, for $s \in (t,1]$, $\lambda_n(s) \P_{\ip{ns}} = \lambda_n(t) \P_{\ip{nt}} + \{ \lambda_n(s) - \lambda_n(t) \} \P_{\ip{ns} - \ip{nt}}^{\star,\ip{ns}}$, where $\P_{\ip{ns} - \ip{nt}}^{\star,\ip{ns}} = (\ip{ns} - \ip{nt})^{-1} \sum_{i=\ip{nt}+1}^{\ip{ns}} \delta_{X_i}$, the first supremum is smaller than
\begin{multline*}
\left\| \lambda_n(t) \left\{ \P_{\ip{nt}}f - P_1 f \right\} \right\|_\FF + \sup_{s \in (t,1]} \left\| \{ \lambda_n(s) - \lambda_n(t) \} \left\{ \P_{\ip{ns} - \ip{nt}}^{\star,\ip{ns}} f - P_2 f \right\} \right\|_\FF \\ + | \lambda_n(t) - t| \left\{ \left\|  P_1 f \right\|_\FF + \left\|  P_2 f \right\|_\FF \right\} + \sup_{s \in (t,1]}  |\lambda_n(s) - s | \times \left\| P_2 f \right\|_\FF
\end{multline*}
and converges to zero in outer probability because $\sqrt{n}  \| \lambda_n(t) \{ \P_{\ip{nt}}f - P_1 f \} \|_\FF = \| \Z_n(t,f) \|_\FF$ and $\sqrt{n} \sup_{s \in (t,1]} \| \{ \lambda_n(s) - \lambda_n(t) \} \{ \P_{\ip{ns} - \ip{nt}}^{\star,\ip{ns}} f - P_2 f \} \|_\FF$ converge in distribution.
\end{proof}


\begin{proof}[\bf Proof of Theorem \ref{thm_H1}~(ii).] 
Let us first show that $\sup_{s \in [0,1]} \| \widehat{\Z}_n(s,f) \|_\FF$ is bounded in outer probability. Since
$$
\sup_{s \in [0,1]} \left\| \widehat{\Z}_n(s,f) \right\|_\FF = \max \left\{ \sup_{s \in [0,t]} \left\| \widehat{\Z}_n(s,f) \right\|_\FF , \sup_{s \in (t,1]} \left\| \widehat{\Z}_n(s,f) \right\|_\FF \right\},
$$
it is sufficient to verify the claim for each of the suprema in the maximum. From the definition of $\widehat{\Z}_n$, under $H_1$, the first supremum converges weakly from the continuous mapping theorem and is therefore bounded in outer probability. For the second supremum, we write
$$
\sup_{s \in (t,1]} \left\| \widehat{\Z}_n(s,f) \right\|_\FF = \max_{\ip{nt}< k\le n} \left\| \widehat{\Z}_n \left( \frac{k}{n},f \right) \right\|_\FF \leq A_n + B_n + C_n + D_n,
$$
where 
\begin{gather*}
A_n = \frac{1}{\sqrt{n}} \left\| \sum_{i=1}^{\ip{nt}} \xi_i f(X_i) \right\|_{\FF}, \qquad B_n =  \max_{\ip{nt}< k\le n} \frac{1}{\sqrt{n}} \left\| \sum_{i=\ip{nt}+1}^{k} \xi_i f(X_i) \right\|_{\FF}\\
C_n =  \max_{\ip{nt}< k\le n} \frac{1}{\sqrt{n}} \left\| \left( \sum_{i=1}^{k} \xi_i \right) \times \frac{1}{k}\sum_{j=1}^{\ip{nt}}f(X_j) \right\|_{\FF} ,
\end{gather*}
and
$$
D_n =  \max_{\ip{nt}< k\le n} \frac{1}{\sqrt{n}} \left\| \left( \sum_{i=1}^{k} \xi_i \right)\times  \frac{1}{k}\sum_{j=\ip{nt}+1}^{k} f(X_j) \right\|_{\FF}.
$$
The quantity $A_n$ is clearly bounded in outer probability by Lemma \ref{lem:boundprob_1} with $t_1=t_2=t$. For $B_n$, we can write
\begin{multline*}
B_n =  \max_{1 \le k \le n - \ip{nt}} \frac{1}{\sqrt{n}} \left\| \sum_{i=1}^k \xi_{i+\ip{nt}} f(X_{i+\ip{nt}}) \right\|_{\FF} \leq \sup_{s \in [0,1-t]}  \frac{1}{\sqrt{n}} \left\| \sum_{i=1}^{\ip{ns}} \xi_{i+\ip{nt}} f(X_{i+\ip{nt}}) \right\|_{\FF} \\ + \frac{1}{\sqrt{n}} \left\| \sum_{i=1}^{\ip{n(1-t)}} \xi_{i+\ip{nt}} f(X_{i+\ip{nt}}) \right\|_{\FF} + \frac{1}{\sqrt{n}} |\xi_n| F_e(X_n),
\end{multline*}
where the two last terms on the right come from the fact that $0 \leq n - \ip{nt} - \ip{n(1-t)} \leq 1$. The first two terms on the right are bounded in outer probability by Lemma \ref{lem:boundprob_1}, and so is the third one because it converges almost surely to zero. It follows that $B_n$ is bounded in outer probability. Now, $C_n$ is bounded above by 
\begin{multline*}
\max_{\ip{nt}< k\le n} \left\{ \frac{1}{\sqrt{n}} \left| \sum_{i=1}^{k} \xi_i \right|  \left\| \frac{1}{k}\sum_{j=1}^{\ip{nt}}f(X_j) \right\|_{\FF} \right\} \leq  \max_{\ip{nt}< k\le n} \left\{ \frac{1}{\sqrt{n}} \left| \sum_{i=1}^{k} \xi_i \right|  \times \frac{1}{k}\sum_{j=1}^{\ip{nt}} F_e(X_j) \right\} \\ \leq \max_{\ip{nt}< k\le n} \frac{1}{\sqrt{n}} \left| \sum_{i=1}^{k} \xi_i \right|  \times \max_{\ip{nt}< k\le n} \frac{1}{k}\sum_{j=1}^{\ip{nt}} F_e(X_j) \leq  I_n,
\end{multline*}
where 
$$
I_n = \sup_{s \in [0,1]} \left\{ \frac{1}{\sqrt{n}} \left| \sum_{i=1}^{\ip{ns}} \xi_i \right| \right\} \times \max_{1\le k\le n} \frac{1}{k}\sum_{j=1}^{k} F_e(X_j).
$$
Similarly, we have that $D_n \leq I_n$. To show that $C_n$ and $D_n$ are bounded in outer probability, we will show that $I_n$ is bounded in probability. 

Since $n^{-1} \sum_{j=1}^{n} F_e(X_j)$ converges to $tP_1F_e+(1-t)P_2F_e<\infty$ almost surely, we have that $\sup_{n\in \N} n^{-1} \sum_{j=1}^{n} F_e(X_j)$ is an almost surely finite random variable. The fact that $I_n$ is bounded in probability then follows from the weak convergence of $s \mapsto n^{-1/2} \sum_{i=1}^{\ip{ns}} \xi_i$ to Brownian motion in $\ell^\infty([0,1])$, which implies that $\sup_{s \in [0,1]} | n^{-1/2} \sum_{i=1}^{\ip{ns}} \xi_i |$ is uniformly tight. 

Thus, $A_n$, $B_n$, $C_n$, and $D_n$ are all bounded in outer probability, which implies that $\sup_{s \in (t,1]} \| \widehat{\Z}_n(s,f) \|_\FF$ is bounded in outer probability, and therefore that $\sup_{s \in [0,1]} \| \widehat{\Z}_n(s,f) \|_\FF$ is bounded in outer probability.

The analogous result for the process $\widehat{\Z}_n^\star$ follows from the fact that $\sup_{s \in [0,1]} \| \widehat{\Z}_n^\star(s,f) \|_\FF$ (written for instance as a maximum over $1 \leq \ip{ns} \leq n$) is nothing else than the version of $\sup_{s \in [0,1]} \| \widehat{\Z}_n(s,f) \|_\FF$ computed from the sequence $(\xi_n,X_n),(\xi_{n-1},X_{n-1}),\dots,(\xi_1,X_1)$ which has the same distribution as the original sequence. The desired result is finally an immediate consequence of~(\ref{hatDn}).
\end{proof}


\begin{proof}[\bf Proof of Theorem \ref{thm_H1}~(iii).] For any $s \in [0,t]$ and $f \in \FF$, we can write
$$
\widecheck{\Z}_n(s,f) = \frac{1}{\sqrt{n}} \sum_{i=1}^{\ip{ns}} \xi_i \{f(X_i) - P_1 f \} -  \{ \P_n f - P_1 f \} \times \frac{1}{\sqrt{n}} \sum_{i=1}^{\ip{ns}} \xi_i,
$$
while, for any $s \in [t,1]$ and $f \in \FF$, we have
\begin{multline}
\label{decomp2}
\widecheck{\Z}_n(s,f) =  \frac{1}{\sqrt{n}} \sum_{i=1}^{\ip{nt}} \xi_i \{f(X_i) - P_1 f \} -  \{ \P_n f - P_1 f \} \times \frac{1}{\sqrt{n}} \sum_{i=1}^{\ip{nt}} \xi_i \\ + \frac{1}{\sqrt{n}} \sum_{i=\ip{nt} +1}^{\ip{ns}} \xi_i \{f(X_i) - P_2 f \} -  \{ \P_n f - P_2 f \} \times \frac{1}{\sqrt{n}} \sum_{i=\ip{nt}+1}^{\ip{ns}} \xi_i.
\end{multline}
Now, let us show that
\begin{equation}
\label{one}
\left( (s,f) \mapsto \frac{1}{\sqrt{n}} \sum_{i=1}^{\ip{ns}} \xi_i \{f(X_i) - P_1 f \}, (s,f) \mapsto \frac{1}{\sqrt{n}} \sum_{i=1}^{\ip{ns}} \xi_i \right) \leadsto \left(\Z_{P_1},(s,f) \mapsto B_1(s)\right)
\end{equation}
in $\{\ell^\infty([0,t] \times \FF)\}^2$, where $B_1$ is Brownian motion and is independent of $\Z_{P_1}$. Using the multivariate central limit theorem, it can be verified that we have weak convergence of the finite-dimensional distributions. Joint asymptotic tightness follows from the weak convergence of $s \mapsto n^{-1/2} \sum_{i=1}^{\ip{ns}} \xi_i$ to $B_1$ in $\ell^\infty([0,1])$, and the weak convergence of $\widetilde{\Z}_n$ to $\Z_{P_1}$ in $\ell^\infty([0,t] \times \FF)$. Since the two component processes on the left of~(\ref{one}) are uncorrelated, their weak limits are independent. Similarly, we have that 
\begin{multline}
\label{two}
\left((s,f) \mapsto \frac{1}{\sqrt{n}} \sum_{i=\ip{nt} +1}^{\ip{ns}} \xi_i \{f(X_i) - P_2 f \}, (s,f) \mapsto \frac{1}{\sqrt{n}} \sum_{i=\ip{nt} +1}^{\ip{ns}} \xi_i \right) \\ \leadsto \left( (s,f) \mapsto \Z_{P_2}(s-t,f),(s,f) \mapsto B_2(s-t) \right)
\end{multline}
in $\{\ell^\infty([t,1] \times \FF)\}^2$, where $B_2$ is Brownian motion independent of $\Z_{P_2}$, and $\left(\Z_{P_2},B_2\right)$ is independent of $\left(\Z_{P_1},B_1\right)$.

Combining the fact that $\P_n f$ converges in outer probability to $t P_1 f + (1-t) P_2 f$ uniformly in $f \in \FF$ with~(\ref{one}), we obtain from the continuous mapping theorem that, in $\ell^\infty([0,t] \times \FF)$, $\widecheck{\Z}_n$ converges weakly to $(s,f) \mapsto \W_{P_1}(s,f) = \Z_{P_1}(s,f) - (1-t) (P_2 f - P_1 f) B_1(s)$. Using~(\ref{two}) similarly, it can be verified that, in $\ell^\infty([t,1] \times \FF)$,
$$
(s,f) \mapsto \frac{1}{\sqrt{n}} \sum_{i=\ip{nt} +1}^{\ip{ns}} \xi_i \{f(X_i) - P_2 f \} -  \{ \P_n f - P_2 f \} \times \frac{1}{\sqrt{n}} \sum_{i=\ip{nt}+1}^{\ip{ns}} \xi_i
$$
converges weakly to $(s,f) \mapsto \W_{P_2}(s,f) = \Z_{P_2}(s-t,f) - t (P_2 f - P_1 f) B_1(s-t)$. By independence of first two terms in~(\ref{decomp2}) with the last two terms, we then obtain that, in $\ell^\infty([t,1] \times \FF)$, $\widecheck{\Z}_n$ converges weakly  to $(s,f) \mapsto \W_{P_1}(t,f) + \W_{P_2}(s,f)$, which implies that $\widecheck{\Z}_n$ converges weakly in $\ell^\infty([0,1] \times \FF)$. The desired result finally follows from the fact that $\widecheck{\D}_n(s,f) = \widecheck{\Z}_n(s,f) - \lambda_n(s) \widecheck{\Z}_n(1,f)$ and the continuous mapping theorem. 
\end{proof}

\subsection{Proofs of Propositions~\ref{prop_S_T_H0} and~\ref{prop_S_T_H1}}

We state a lemma before giving the proofs of the propositions.

Let $A$ be the space of bounded Borel measurable functions on $\R^d$ and let $B$ be the space of c.d.f.s of Borel probability measures on $\R^d$. The spaces $A$ and $B$ are subsets of $\ell^\infty(\Rbar^d)$ and the topologies on $A$ and $B$ are the ones induced by uniform convergence. The following result (the help of Johan Segers is gratefully acknowledged) will allow us to apply the continuous mapping theorem in the proof of Proposition~\ref{prop_S_T_H0}. 

\begin{lem}
\label{phi}
Let $\phi:A \times B \to \R$ be defined by $\phi(a,b) = \int_{\R^d} a \dd b$. The map $\phi$ is continuous at each $(a,b) \in A \times B$ such that $a$ is continuous on $\R^d$.
\end{lem}

\begin{proof} 
Let $(a_n, b_n)$ be a sequence in $A \times B$ such that $\sup_{x \in \R^d} |a_n(x) - a(x)| \to 0$ and $\sup_{x \in \R^d} |b_n(x) - b(x)| \to 0$. It is sufficient to show that $\int_{\R^d} a_n \, \dd b_n \to \int_{\R^d} a \, \dd b$. By the triangle inequality,
$$
\left| \int_{\R^d} a_n \, \dd b_n - \int_{\R^d} a \, \dd b \right|
\le \left| \int_{\R^d} a_n \, \dd b_n - \int_{\R^d} a \, \dd b_n \right|
+ \left| \int_{\R^d} a \, \dd b_n - \int_{\R^d} a \, \dd b \right|.
$$
For the first term on the right of this inequality, we have
\begin{multline*}
  \left| \int_{\R^d} a_n \, \dd b_n - \int_{\R^d} a \, \dd b_n \right| \le \int_{\R^d} | a_n - a | \, \dd b_n \\ \le \sup_{x \in \R^d} |a_n(x) - a(x)|  \int_{\R^d} \dd b_n = \sup_{x \in \R^d} |a_n(x) - a(x)| \to 0. 
\end{multline*}
For the second term, since $\sup_{x \in \R^d} |b_n(x) - b(x)| \to 0$, we have that $b_n(x) \to b(x)$ for every $x \in \R^d$, which, by the Portmanteau lemma 
and the continuity of the function $a$ implies that $\int_{\R^d} a \, \dd b_n \to \int_{\R^d} a \, \dd b$.
\end{proof}

\begin{proof}[\bf Proof of Proposition~\ref{prop_S_T_H0}.] A consequence of Theorem~\ref{thm_H0} and the fact that, under $H_0$, $F_n$ converges almost surely to $F_0$ (the c.d.f.~of $P_0$) uniformly in $x \in \R^d$ is that 
$$
\left( \D_n,\widehat{\D}_n^{(1)},\dots,\widehat{\D}_n^{(N)},\widecheck{\D}_n^{(1)},\dots,\widecheck{\D}_n^{(N)},F_n \right) \leadsto \left( \D_{P_0},\D_{P_0}^{(1)},\dots,\D_{P_0}^{(N)},\D_{P_0}^{(1)},\dots,\D_{P_0}^{(N)},F_0 \right)
$$ 
in $\{ \ell^\infty([0,1] \times \Rbar^d) \}^{(2N+2)}$. Using the map $\phi$ defined in Lemma~\ref{phi}, it is easy to see that $S_{n,\vee} = \sup_{s\in[0,1]} \phi [ \{ \D_n (s,\cdot) \}^2,F_n ]$, that $\widehat{S}_{n,\vee}^{(j)} = \sup_{s\in[0,1]} \phi [\{ \widehat{\D}_n^{(j)} (s,\cdot) \}^2,F_n ]$ and that $\widecheck{S}_{n,\vee}^{(j)} = \sup_{s\in[0,1]} \phi [ \{ \widecheck{\D}_n^{(j)} (s,\cdot) \}^2,F_n ]$, $j \in \{1,\dots,N\}$. Furthermore, the limiting process $\D_{P_0}$ is continuous almost surely. The result then follows from Lemma~\ref{phi} and the continuous mapping theorem. 
\end{proof}

\begin{proof}[\bf Proof of Proposition~\ref{prop_S_T_H1}.]
From Theorem~\ref{thm_H1}, we have that
$$
\sup_{s \in [0,1]} \sup_{x \in \R^d} \left| \frac{\D_n(s,x)}{\sqrt{n}} - K_t(s,x) \right| \p 0,
$$
where $K_t(s,x)= \{ F_1(x)-F_2(x) \} (s\wedge t) \{ 1- (s\vee t) \}$, i.e., $\D_n/\sqrt{n} \p K_t$ in $\ell^\infty([0,1] \times \Rbar^d)$. Also, under $H_1$, $F_n$ converges almost surely to $F_t = t F_1 + (1-t) F_2$ uniformly in $x \in \R^d$. Then, 
\begin{equation}
\label{ineq}
\sup_{s \in [0,1]} \left|  \int_{\R^d} \frac{\{\D_n(s,x)\}^2}{n} \dd F_n(x) - \int_{\R^d} \{K_t(s,x)\}^2 \dd F_t(x) \right| \leq A_n + B_n,
\end{equation}
where
$$
A_n = \sup_{s \in [0,1]} \left| \int_{\R^d} \frac{\{\D_n(s,x)\}^2}{n} \dd F_n(x) - \int_{\R^d} \{K_t(s,x)\}^2 \dd F_n(x) \right| 
$$
and
$$
B_n =  \sup_{s \in [0,1]} \left| \int_{\R^d} \{K_t(s,x)\}^2 \dd F_n(x) - \int_{\R^d} \{K_t(s,x)\}^2 \dd F_t(x) \right|.
$$
We have 
$$
A_n \leq \sup_{s \in [0,1]} \sup_{x \in \R^d} \left| \frac{\{\D_n(s,x)\}^2}{n} - \{K_t(s,x)\}^2 \right| \p 0
$$
by the continuous mapping theorem, and, with the notation $g_s = \{ K_t(s,\cdot) \}^2$, $s \in [0,1]$,
\begin{multline*}
B_n 
= \sup_{s \in [0,1]} \left| \lambda_n(t) \P_{\ip{nt}} g_s + \{ 1-\lambda_n(t) \}\P_{n-\ip{nt}}^\star g_s - t P_1 g_s - (1-t) P_2 g_s \} \right| \\ \leq \lambda_n(t) \sup_{s \in [0,1]} \left| (\P_{\ip{nt}} - P_1) g_s \right| + \{1 - \lambda_n(t) \} \sup_{s \in [0,1]} \left| (\P_{n - \ip{nt}}^\star - P_2 ) g_s \right| + 2 |\lambda_n(t) - t|
\end{multline*}
because $\sup_{s \in [0,1]} | P_1 g_s| \leq 1$ and $\sup_{s \in [0,1]} | P_2 g_s| \leq 1$. Now, 
$$
 \sup_{s \in [0,1]} \left| (\P_{\ip{nt}} - P_1) g_s \right| \leq \left| (\P_{\ip{nt}} - P_1) (F_1 - F_2)^2 \right| \p 0
$$
and
$$
\sup_{s \in [0,1]} \left| (\P_{n - \ip{nt}}^\star - P_2 ) g_s \right| \leq \left| (\P_{n - \ip{nt}}^\star - P_2 ) (F_1 - F_2)^2 \right| \p 0
$$
by the law of large numbers, which implies that $B_n \p 0$. It follows from~(\ref{ineq}) that 
$$
s \mapsto \int_{\R^d} \frac{\{\D_n(s,x)\}^2}{n} \dd F_n(x) \quad  \p  \quad s \mapsto \int_{\R^d} \{K_t(s,x)\}^2 \dd F_t(x)
$$
in $\ell^\infty([0,1])$, from which, using the continuous mapping theorem, we obtain that 
$$
S_{n,\vee}/n \p \sup_{s \in [0,1]} \int_{\R^d} \{ K_t(s,x) \}^2 \dd F_t(x).
$$
Since $F_1$ and $F_2$ are distinct and $K_t(s,x)= \{ F_1(x)-F_2(x) \} (s\wedge t)\{1- (s\vee t)\}$, we have that $\sup_{s \in [0,1]} \int_{\R^d} \{ K_t(s,x) \}^2 \dd F_t(x) > 0$, which implies that $S_{n,\vee} \p +\infty$.

Now, let $j \in \{1,\dots,N\}$. The claim for $\widehat{S}_{n,\vee}^{(j)}$ follows from the inequality
$$
\widehat{S}_{n,\vee}^{(j)} = \sup_{s \in [0,1]} \int_{\R^d} \{ \widehat{\D}_n^{(j)}(s,x)\}^2 \dd F_n(x) \leq \sup_{s \in [0,1]} \sup_{y \in \R^d} \{ \widehat{\D}_n^{(j)}(s,y)\}^2 \int_{\R^d} \dd F_n(x) \leq \{T_{n,\vee}^{(j)}\}^2,
$$
and the fact that $\widehat{T}_{n,\vee}^{(j)}$ is bounded in outer probability from Theorem~\ref{thm_H1}. It remains finally to prove the claim for $\widecheck{S}_{n,\vee}^{(j)}$. A consequence of Theorem~\ref{thm_H1} is that $(\widecheck{\D}_n^{(j)},F_n)$ converges weakly in $\{ \ell^\infty([0,1] \times \Rbar^d) \}^2$, which, combined with the fact that $\widecheck{S}_{n,\vee}^{(j)} = \sup_{s\in[0,1]} \phi [ \{ \widecheck{\D}_n^{(j)} (s,\cdot) \}^2,F_n ]$ and Lemma~\ref{phi}, implies that $\widecheck{S}_{n,\vee}^{(j)}$ converges in distribution, and hence, that it is bounded in outer probability.
\end{proof}

\section{Implementation of the tests based on $U_{n,k}$ and $V_{n,k}$}
\label{implementation}

Let $m > 2$ be an integer and let $a_1,\dots,a_m$ be elements of $\Smc_d^+$ uniformly spaced over $\Smc_d^+$. For $k \in \{1,\dots,n-1\}$, the following numerical approximations are then considered:
\begin{align*}
U_{n,k} &\approx \frac{1}{m} \sum_{l=1}^m \frac{1}{n} \sum_{q=1}^n \left\{ \D_n \left( \frac{k}{n},a_l,a_l^\top X_q \right) \right\}^2 \\ 
&= \frac{k^2 (n - k)^2}{n^4 m} \sum_{l=1}^m \sum_{q=1}^n \left\{ \frac{1}{k} \sum_{i=1}^k \1(a_l^\top X_i \leq a_l^\top X_q) - \frac{1}{n-k} \sum_{i=k+1}^n \1(a_l^\top X_i \leq a_l^\top X_q) \right\}^2, 
\end{align*}
and
\begin{align*}
V_{n,k} &\approx \max_{1 \leq l \leq m} \max_{1 \leq q \leq n} \left|  \D_n \left( \frac{k}{n},a_l,a_l^\top X_q \right) \right| \\
&= \frac{k (n - k)}{n^{3/2}} \max_{1 \leq l \leq m} \max_{1 \leq q \leq n} \left|  \frac{1}{k} \sum_{i=1}^k \1(a_l^\top X_i \leq a_l^\top X_q) - \frac{1}{n-k} \sum_{i=k+1}^n \1(a_l^\top X_i \leq a_l^\top X_q) \right|.
\end{align*}

Next, recall that, for any $j \in \{1,\dots,N\}$, $\widehat{\D}_n^{(j)}$ (resp.\ $\widecheck{\D}_n^{(j)}$) is the version of $\widehat{\D}_n$ (resp.\ $\widecheck{\D}_n$) computed from $\xi_1^{(j)},\dots,\xi_n^{(j)}$. Proceeding as above, for any $j \in \{1,\dots,N\}$, the multiplier versions of $U_{n,k}$ and $V_{n,k}$ based on the process $\widehat{\D}_n^{(j)}$ are computed respectively as
\begin{multline*}
\widehat{U}_{n,k}^{(j)} \approx \frac{1}{mn^4}\sum_{l=1}^m \sum_{q=1}^n \left\{ (n-k) \sum_{i=1}^k (\xi_i^{(j)} - \bar \xi_k^{(j)}) \1(a_l^\top X_i \leq a_l^\top X_q) \right. \\ - \left.  k \sum_{i=k+1}^n (\xi_i^{(j)} - \bar \xi_{n-k}^{(j)}) \1(a_l^\top X_i \leq a_l^\top X_q) \right\}^2, 
\end{multline*}
and
\begin{multline*}
\widehat{V}_{n,k}^{(j)} \approx n^{-3/2} \max_{1 \leq l \leq m} \max_{1 \leq q \leq n} \left|  (n-k) \sum_{i=1}^k (\xi_i^{(j)} - \bar \xi_k^{(j)}) \1(a_l^\top X_i \leq a_l^\top X_q) \right. \\ - \left. k \sum_{i=k+1}^n (\xi_i^{(j)} - \bar \xi_{n-k}^{(j)}) \1(a_l^\top X_i \leq a_l^\top X_q) \right|.
\end{multline*}
Similarly, for any $j \in \{1,\dots,N\}$, the multiplier versions of $U_{n,k}$ and $V_{n,k}$ based on the process $\widecheck{\D}_n^{(j)}$ are computed respectively as
\begin{multline*}
\widecheck{U}_{n,k}^{(j)} \approx  \frac{1}{mn^4}\sum_{l=1}^m \sum_{q=1}^n \left[ (n-k) \sum_{i=1}^k \xi_i^{(j)} \{ \1(a_l^\top X_i \leq a_l^\top X_q) - F_{a_l,n}(a_l^\top X_q) \} \right. \\ \left. - k \sum_{i=k+1}^n \xi_i^{(j)} \{ \1(a_l^\top X_i \leq a_l^\top X_q) - F_{a_l,n}(a_l^\top X_q) \} \right]^2, 
\end{multline*}
and
\begin{multline*}
\widecheck{V}_{n,k}^{(j)} \approx n^{-3/2} \max_{1 \leq l \leq m} \max_{1 \leq q \leq n} \left| (n-k) \sum_{i=1}^k \xi_i^{(j)} \{ \1(a_l^\top X_i \leq a_l^\top X_q) - F_{a_l,n}(a_l^\top X_q) \} \right. \\- \left. k \sum_{i=k+1}^n \xi_i^{(j)} \{ \1(a_l^\top X_i \leq a_l^\top X_q) - F_{a_l,n}(a_l^\top X_q) \} \right|,
\end{multline*}
where $F_{a,n}$ is the empirical c.d.f.\ computed from the projected sample $a^\top X_1,\dots,a^\top X_n$.

\bibliographystyle{plainnat}
\bibliography{biblio}

\end{document}